\documentclass[submission,copyright,creativecommons]{eptcs}
\providecommand{\event}{AFL 2026} % Name of the event you are submitting to

\usepackage{iftex}

\ifpdf
  \usepackage{underscore}         % Only needed if you use pdflatex.
  \usepackage[T1]{fontenc}        % Recommended with pdflatex
\else
  \usepackage{breakurl}           % Not needed if you use pdflatex only.
\fi

\newcommand{\longversion}[1]{}
\newcommand{\shortversion}[1]{#1}
\usepackage{subcaption}        
\usepackage{hyperref}

\usepackage[utf8]{inputenc}

\usepackage[]{todonotes}
\usepackage{hyperref} 
\usepackage{mathtools} 
\usepackage{amsmath,amsfonts,amssymb} 

\usepackage[capitalize]{cleveref}
\usepackage{complexity}
\usepackage{xspace}

\usetikzlibrary{positioning, shapes,patterns,shadows.blur, arrows,decorations,arrows,automata,shadows,patterns,chains,graphs,calc,intersections,matrix,fit,shapes,chains,decorations.pathreplacing,chains,fit,shapes}

\DeclareMathOperator{\classical}{cl}
\DeclareSymbolFont{Shuffle}{U}{shuffle}{m}{n}
\DeclareFontFamily{U}{shuffle}{}
\DeclareFontShape{U}{shuffle}{m}{n}{%
  <-8>shuffle7%
  <8->shuffle10%
}{}
\DeclareMathSymbol\shuffle{\mathbin}{Shuffle}{"001}
\DeclareMathSymbol\cshuffle{\mathbin}{Shuffle}{"002}

\newcommand{\emptyword}{\lambda}

\def\ta{\mathtt{a}}
\def\tb{\mathtt{b}}
\def\tc{\mathtt{c}}
\def\td{\mathtt{d}}
\def\te{\mathtt{e}}
\def\tf{\mathtt{f}}

\def\tn{\mathtt{n}}

\DeclareMathOperator{\alphabet}{alph}

\newcommand{\N}{\mathbb{N}}

\newcommand{\cG}{\mathcal{G}}

\renewcommand{\alph}{\alphabet}

\def\bico{\mathtt{\mathrm{bico}\text{\,-\,}}}

\newcommand{\iffl}{\shortversion{iff}\longversion{if and only if}\xspace}

\longversion{\usepackage[appendix=inline]{apxproof}}
\shortversion{\usepackage[appendix=strip]{apxproof}}
\shortversion{}

\newcommand{\todoscs}[1]{\todo[color=green!50,inline]{#1}}
\newcommand{\todoscsInPlace}[1]{\todo[color=green!50]{#1}}
\newcommand{\todohf}[1]{\todo[color=red!80,inline]{#1}}
\newcommand{\todohfInPlace}[1]{\todo[color=red!80]{#1}}

\shortversion{\renewcommand{\iff}{\Leftrightarrow}}

\newcommand{\trd}{\textsf{td}}

\newcommand{\nd}{\textsf{nd}}

\newif\ifinappendix

\newcommand{\applabel}[1]{
\shortversion{\ifinappendix\hypertarget{app:#1}{}\else
{$(*)$}\fi}
\label{#1}}

\newcommand{\intervals}{\mathcal{I}}

\usepackage{comment}
\usepackage[T1]{fontenc}
\usepackage{graphicx}
\usepackage{xcolor}
\usepackage{afterpage}
\definecolor{LB}{RGB}{188,218,234}
\definecolor{DB}{RGB}{33,121,180}
\definecolor{LR}{RGB}{252,179,179}
\definecolor{DR}{RGB}{233,76,78}
\definecolor{DG}{RGB}{51,159,43}
\definecolor{DO}{RGB}{255,126,0}
\usepackage{enumitem}
\newtheorem{theorem}{Theorem}
\newtheorem{lemma}{Lemma}
\newtheorem{proposition}{Proposition}
\newtheorem{corollary}{Corollary}

\newtheorem{remark}{Remark}
\newtheorem{definition}{Definition}
\newtheorem{example}{Example}

\newtheoremrep{thm}[theorem]{Theorem}
\newtheoremrep{lem}[lemma]{Lemma}
\newtheoremrep{obs}[observation]{Observation}
\newtheoremrep{cor}[corollary]{Corollary}
\newtheoremrep{prp}[proposition]{Proposition}

\title{Language-Representability:\\ Possibilities and Limitations}
\author{Zhidan Feng
\institute{BTU, China} \email{zhidanfeng@bjut.edu.cn}
\and
Pamela Fleischmann \institute{Kiel University, Germany}\email{fpa@informatik.uni-kiel.de}
\and
Henning Fernau \qquad Kevin Mann \qquad Silas Cato Sacher
\institute{Trier University, Germany} \email{\{fernau,mann,sacher\}@uni-trier.de}
}
\def\titlerunning{Language-Representability: Possibilities and Limitations}
\def\authorrunning{Z. Feng, H. Fernau, P. Fleischmann, K. Mann, 
S.~C. Sacher}
\begin{document}
\maketitle
%\begin{document}
%
%\author{\inst{1}\orcidID{0000-0002-3364-5396} \and
%\inst{2}\orcidID{0000-0002-4444-3220} \and
%\inst{3}\orcidID{0000-0002-1531-7970} \and 
%\inst{2}\orcidID{0000-0002-0880-2513} \and
%\inst{2}\orcidID{0009-0004-6850-1298}
%}
%
%\longversion{}
% First names are abbreviated in the running head.
% If there are more than two authors, 'et al.' is used.
%
%\longversion{ %}}
%

\begin{comment}
\author{Zhidan Feng}
{Universit\"at Trier, Fachbereich IV, Informatikwissenschaften, Germany \and Shandong University, School of Mathematics and Statistics, Weihai, China}
{zhidanfeng@mail.sdu.edu.cn}
{THTps://orcid.org/0000-0002-3364-5396}
{}

\author{Henning Fernau}
{Universit\"at Trier, Fachbereich IV, Informatikwissenschaften, Germany \and \url{THTps://www.uni-trier.de/index.php?id=49861}}
{fernau@uni-trier.de}
{THTps://orcid.org/0000-0002-4444-3220}
{}

\author{Pamela Fleischmann}
{Kiel University, Germany}
{fpa@informatik.uni-kiel.de}
{THTps://orcid.org/0000-0002-1531-7970} 
{}

\author{Kevin Mann}
{Universit\"at Trier, Fachbereich IV, Informatikwissenschaften, Germany}
{mann@uni-trier.de}
{THTps://orcid.org/0000-0002-0880-2513} 
{}

\author{Silas Cato Sacher}
{Universit\"at Trier, Fachbereich IV, Informatikwissenschaften, Germany}
{sacher@informatik.uni-trier.de}
{THTps://orcid.org/0009-0004-6850-1298} 
{}
\end{comment}

%\maketitle              % typeset the header of the contribution
%

\begin{abstract}
The study of word-representability was initiated by the seminal work of Kitaev and Pyatkin in 2008 that has later led to the monograph by Kitaev and Lozin in 2015. In this paper, we build on the very recent work by Fernau \emph{et al.} who proposed  a general framework that generalizes certain aspects of word-representability, so that any binary language describes a graph class. In this work, we systematically study particularly small languages and observe that they characterize well-known graph classes, e.g., interval, permutation, circle, and bipartite chain graphs. Thus, we strengthen the bond between formal languages and graph classes, even for small binary languages. We also show some limitations of our approach by proving that, e.g., families of sparse graphs like planar graphs cannot be characterized by any language following this approach. 
%\keywords{Graph classes \and Word-representable graphs \and Language-representability \and Twins in graphs}
\end{abstract}

\setcounter{footnote}{0}

\section{Introduction}

In classical word-representability, a word $w$ represents a graph $G$ iff the graph's vertices are exactly the word's letters and
two letters in $w$ alternate iff they are adjacent in $G$ \cite{KitLoz2015}. This pattern of alternation has been altered since then to overcome the disadvantage that not all graphs are word-representable in the classical way (e.g., $k$-11-representability \cite{CheKKKP2019}, permutational representability \cite{KitSei2008}).
Also, word-representability has been generalized towards so-called multi-word-representability in \cite{KenMal2023,KenSalSas2025}; this approach is different from ours but we also introduce a generalization of both approaches. 
In this paper, we investigate further the recent notion of $L$-representability for a formal language $L$ \cite{FerFMS2026a}. This notion allows for other patterns than alternations; in particular  a binary language $L$ represents a class of graphs \iffl there exists a word $w\in L$
such that the projection of $w$ onto two letters follows the pattern given by~$L$ \iffl the associated vertices are adjacent. With this new notion, we are able to change the pattern, e.g., to the language of palindromes, to Lyndon words, or to $\{0^n1^n,1^n0^n|\,n\geq 1\}$ and other variants of the Dyck language, see \cite{FerFMS2026a}, or to any other language with a certain symmetry property.
Consider for instance the regular language $L_{\text{even}}\coloneqq\{w\in\{0,1\}^{\ast}|\,|w|_0\equiv_2 0 \vee |w|_1\equiv_20\}$, i.e., the language of all words having an even number of $0$s or $1$s. 
Additionally, let $w\in V^{\ast}$ be a word over the alphabet $V$ and $v_1,v_2\in V$. Then $w$ represents the following graph: if $|v_1|_0\equiv_2|v_2|_0\equiv_20$ then there is an edge between $v_1$ and $v_2$; if $|v_1|_0\equiv_2|v_2|_0\equiv_21$ then there is no edge between $v_1$ and $v_2$; lastly every $v$ with $|v|_0\equiv_20$ is connected to every $v'$ with $|v'|_0\equiv_2 1$. Thus, $L_{\text{even}}$ represents graphs where some vertices form a clique, the others are an independent set, and all vertices from the clique are adjacent to all vertices from the independent set - a split graph.

One of the natural syntactic complexity measures \shortversion{w.r.t.}\longversion{with respect to} word-representable graphs is the \emph{representation number} $\mathcal{R}(G)$ of a graph~$G$ which is the smallest $k \geq 1$ such that there is a $k$-uniform word-representing~$G$.
For instance, it is known that a graph is a circle graph \iffl $\mathcal{R}(G)\leq 2$ \cite{EnrKit2019,Kit2017}. This corresponds in our setting to $L=\{01,10,010,101,0101,1010\}$ characterizing circle graphs.
Thus, studying languages of limited word size already gives insights into graph classes in our new setting. Unfortunately, not all graph classes are representable in this way, but this downside is compensated by the fact that we can represent or exclude whole graph classes as, e.g., each representable graph class is closed under \emph{twinning}, i.e., all classes not closed under twinning cannot be representable in our setting and thus can be excluded immediately without any further proof.
This excludes all sparse graph classes  (like planar graphs) and complements thereof.

%\todo[inline]{Is this true?? Think of complementation!}

%\todo[inline]{You can cite \cite{FerFMS2026a,FerFMS2026} as you wish for the DLT project.}

%Conversely, it is important to study the limitations of a new approach. Can we represent each interesting (hereditary) graph class %by a language? We will see that this is not the case, basically because all representable graph classes have a particular \emph{twinning} closure property that excludes, for instance, characterizing the class of planar graphs, or the class of graphs of treewidth at most~3, etc.

%Presenting these two extremes of language-representability should help put this framework into a proper research context. Of course, many questions remain to be investigated, including possible applications.

\section{Preliminaries}

Let $\N$ denote the natural numbers including $0$, let  $[n]=\{i\in\N \mid 1\leq i\leq n\}$ and  $\N_{\geq n}=\{k\in \N\mid k\geq n\}$ for some $n\in\N$.
%The cardinality of a set $X$ is denoted by $|X|$. 
For a set $X$ and $k\in\N$, define $\binom{X}{k}=\{Y\subseteq X\mid |Y|=k\}$. 

\paragraph*{Language Theory and Combinatorics on Words.}
Each non-empty, finite set $\Sigma$ is an {\em alphabet} whose elements are called {\em letters}.
 A \emph{word} over $\Sigma$ is a finite sequence  of letters from~$\Sigma$; $\Sigma^{\ast}$ denotes the set of all words over~$\Sigma$, including the \emph{empty word}~$\emptyword$. 
Each \longversion{subset $L$ of $\Sigma^{\ast}$}\shortversion{$L\subseteq\Sigma^*$} is called a \emph{language} over~$\Sigma$. Let $\Sigma^+=\Sigma^{\ast}\setminus\{\emptyword\}$.  The {\em length of a word $w$}\longversion{, i.e., the number of its letters,} is denoted by $|w|$. 
For $w\in\Sigma^{\ast}$ and for all $i\in[|w|]$, $w[i]$ denotes the $i^{\text{th}}$ letter of~$w$.
The number of occurrences of $\ta\in\Sigma$ in $w\in\Sigma^*$ is defined as $|w|_{\ta}=|\{i\in[|w|]\mid w[i]=\ta\}|$ and $w$'s alphabet is given by $\alphabet(w)\longversion{=\{\ta\in\Sigma\mid\exists i\in[|w|]:\,w[i]=\ta\}}$. Define $F(w)=\{|w|_a\mid a\in \Sigma\}$ as the 
{\em frequentnesses}.
A word $w\in\Sigma^{\ast}$ is called \emph{$k$-uniform} for some $k\in\N$ if $F(w)=\{k\}$\longversion{, i.e., $|w|_{\ta}=k$ for all $\ta\in\Sigma$}. 
Let $\Sigma^{k\text{-uni}}\coloneqq\{w\in \Sigma^*\mid F(w)=\{k\}\}$ \longversion{collect}\shortversion{be} all \emph{$k$-uniform words} and $\Sigma^{k,\ell\text{-uni}}\coloneqq  \{w\in \Sigma^*\mid F(w)=\{k,\ell\}\}$
all \emph{$(k,\ell)$-uniform words}, i.e., $\Sigma^{k\text{-uni}}=\Sigma^{k,k\text{-uni}}$.
Define the \emph{reversal} $w^R$ of $w\in\Sigma^{\ast}$ by $w^R=w[|w|]\cdot w[|w|-1]\cdots w[1]$.
%;  $w$ is a \emph{palindrome} if $w=w^R$ holds.
%The word $w\in\Sigma^{\ast}$ is called a \emph{repetition} if there exist $u\in\Sigma^{\ast}$ and $k\in\N_{\geq 2}$ such that $w=u^k$ where $u^0=\emptyword$ and $u^k=uu^{k-1}$. Repetitions with $k=2$ are called \emph{copy-words} and words that are not a repetition are called {\em primitive}. \longversion{In addition to the concatenation of words, we need}\shortversion{We also use} the {\em shuffle product}
%to combine two words: for $u,v\in\Sigma^*$, define the {\em shuffle} by $u\shuffle v \coloneqq \{x_1y_1x_2y_2\cdots x_ny_n\mid \exists x_1,\dots,x_n, y_1,\dots,y_n\in\Sigma^*: u=x_1\cdots x_n\land v=y_1\cdots y_n\}$.
%\longversion{As an example, consider the words $\mathtt{ml}$ and $\mathtt{eon}$; we can take the $\mathtt{m}$ from the first word, then the $\mathtt{eo}$ from the second, $\mathtt{l}$ from the first, and finally $\mathtt{n}$ from the second to get $\mathtt{meoln}$ as one word of the shuffle product. But also $\mathtt{melon}$, $\mathtt{emlon}$, and $\mathtt{mleon}$ are examples for words from the shuffle-product of these two words.}
A mapping $f$ from the free monoid $\Sigma^{\ast}$ into another monoid is called a {\em morphism} if
$f(xy)=f(x)f(y)$ holds for all $x,y\in\Sigma^{\ast}$. Note that a morphism is uniquely defined by giving the images of all letters. For $A \subseteq \Sigma$,  the \emph{projective morphism} 
$h_A: \Sigma^* \rightarrow A^*$ is defined by  $h_A(\ta) = \ta$ for $\ta \in A$ and $h_A(\ta) = \emptyword$ otherwise. \longversion{Thus, we have $h_{\ta}(\mathtt{banana})=\ta\ta\ta$.} Let $\widetilde{\cdot}:\{0,1\}^*\rightarrow\{0,1\}^*$ be the \emph{complement morphism} mapping $0$ to $1$ and $1$ to $0$\longversion{. Here we have for instance $\widetilde{010}=101$. Note that the complement morphism is an involution, i.e.,}\shortversion{ which is an involution, i.e., } $\widetilde{\widetilde{w}}=w$.
\longversion{
\bigskip

We finish this section with the main definitions from formal language theory.}
%\shortversion{.}\longversion{, e.g., the set of all Lyndon words is a language.} 
%We extend the concatenation to languages, so that we can define powers of a language~$L$ and the \emph{Kleene star} of $L$ as $L^*=\bigcup_{n\in\N}L^n$ (analogously, we define other operations on words for entire languages). 
We define the \emph{symmetric hull} operator $\langle L\rangle\coloneqq L\cup \widetilde{L}$.   A language $L\subseteq \{0,1\}^{\ast}$ is called {\em $0$-$1$-symmetric} if $L=\widetilde{L}$ (i.e., $L=\langle L \rangle$). For instance, the language $L_{\classical} =(1\cup\emptyword)(01)^*(0\cup\emptyword)$ is $0$-$1$-symmetric. This also holds for its complement $\overline{L_{\classical}}$ where $\overline{L}$ denotes $\Sigma^{\ast}\backslash L$ for some language $L\subseteq\Sigma^{\ast}$. Let $\mathrm{freq}(L)=\{n\in\mathbb{N}_{\geq 1} \mid \exists w\in L: |w|_0 = n\}$ denote the set of frequentnesses in a $0$-$1$-symmetric language~$L\subseteq \{0,1\}^{\ast}$.
%Identifying singleton sets with their elements, we can build \emph{regular expressions} from letters by using concatenation $\cdot$, union $\cup$, and Kleene star. We also include set complementation~$\bar\cdot$ and the \emph{shuffle}~$\shuffle$ when building such expressions. 

%\medskip

\paragraph*{Graph Theory.} 
\longversion{Throughout this paper, w}\shortversion{W}e only consider undirected finite graphs. A graph~$G$ is a pair $(V,E)$ with the finite, non-empty set of vertices $V$ and the set of edges $E\subseteq\binom{V}{2}$. For a given graph~$G$, let $V(G)$ and $E(G)$ denote its sets of vertices and edges, respectively.
%The cardinality of $V$ is called the \emph{order} of~$G$. 
If $\{u,v\}\in E$, $u$ is called a \emph{neighbor} of~$v$, and $N(v)\subseteq V$ are the neighbors of~$v$; $N[v]\coloneqq N(v)\cup\{v\}$ is the \emph{closed neighborhood} of~$v$. 
Two vertices $u,v$ are called \emph{true} resp. \emph{false} \emph{twins} if $N[u]=N[v]$ resp. $N(u)=N(v)$.
A vertex in a graph is called \emph{universal} if $N[v]=V$ and {\em isolated} if $N(v)=\emptyset$.
% neighbor and neighborhood are only used to define degree and universal and isolated vertices
The {\em degree} of \longversion{a vertex }$v\in V$ is\longversion{ defined as} $\deg(v)\coloneqq|N(v)|$. 
% degree is used
The \emph{complement} of the graph $G$, written $\overline{G}$, satisfies $\overline{G}=(V,\binom{V}{2}\setminus E)$. For a graph class $\mathcal{G}$, $\text{co-}\mathcal{G} \coloneqq \{ \overline{G} \mid G \in \mathcal{G}\}$. 
% \text{co-}\mathcal{G} is used
If $G_1=(V_1,E_1)$ and $G_2=(V_2,E_2)$ are graphs, then $\varphi:V_1\to V_2$ 
is a \emph{graph morphism} iff $\{u,v\}\in E_1$ implies $\{\varphi(u),\varphi(v)\}\in E_2$; a bijection $\varphi$ 
is a \emph{graph isomorphism} iff $\{u,v\}\in E_1\iff \{\varphi(u),\varphi(v)\}\in E_2$;\longversion{ if such an isomorphism exists,} then we write $G_1\simeq G_2$.
%graph isomorphism and \simeq are used
A graph $G_1=(V_1,E_1)$ is a \emph{subgraph} of $G_2=(V_2,E_2)$ iff $V_1\subseteq V_2$ and $E_1\subseteq E_2$. $G_1$ is an  \emph{induced subgraph} of $G_2$ ($G_1=G_2[V_1]$) iff $E_1=E_2\cap \binom{V_1}{2}$. If $G_1=(V_1,E_1)$ and $G_2=(V_2,E_2)$ are two graphs with (often, but not necessarily) disjoint sets of vertices $V_1$ and $V_2$, then the (graph) \emph{union} of $G_1$ and $G_2$ is $G_1\cup G_2=(V_1\cup V_2,E_1\cup E_2)$ while the (graph) \emph{join}  of $G_1$ and $G_2$ is
$G_1\nabla G_2=(V_1\cup V_2,E_1\cup E_2\cup \{\{x_1,x_2\}\mid x_1\in V_1, x_2\in V_2\})$.
%subgraph is not used
Graph classes closed under taking induced subgraphs are called \emph{hereditary.}

We call a graph with $n\in\N$ vertices and an empty edge set a {\em null graph}, denoted by $N_n$, and a graph with $n$ vertices and all possible edges a {\em complete graph}, denoted by $K_n$, i.e. $\overline{N_n}=K_n$.
%complete graphs are used
A graph with $n+m$ vertices, with $n,m\in\mathbb{N}_{\geq 1}$ is called a \emph{complete bipartite graph} $K_{n,m}=(V,E)$ if $V=U\cup W$, $U\cap W=\emptyset$, with $|U|=n$ and $|W|=m$ and $E$ contains all edges between $U$ and $W$ but no other edges. 
% complete bipartite graphs are used
%A graph  with $n\in\N$ vertices is a \emph{path} $P_n=(V,E)$ if there is a linear ordering~$<$ on its vertex set such that $\{u,v\}\in E$ if either $u$ is the immediate predecessor of~$v$ in~$<$ or $u$ is the immediate successor of~$v$ in~$<$.
%A graph  with $n\in\N$ vertices is a \emph{cycle} $C_n=(V,E)$ if it contains a path $P_n=(V,E')$ and one additional edge~$e$ that connects the two vertices of this path that have degree one.
%universal vertices are not used
%isoladed vertices are used
% graph union is used 
%while the (graph) \emph{join}  of $G_1$ and $G_2$ is
%$G_1\nabla G_2=(V_1\cup V_2,E_1\cup E_2\cup \{\{x_1,x_2\}\mid x_1\in V_1, x_2\in V_2\})$.
%\longversion{We can describe many graph classes by our framework whose definitions are explained in}\shortversion{For the definition of many graph classes, we refer to} \url{graphclasses.org}.%\todohfInPlace{Here are the graph class definitions hidden.}
Let $G=(V,E)$ be a graph with a partition $V=A\cup B$. If $\binom{A}{2}\cap E=\emptyset$ and $\binom{B}{2}\cap E=\emptyset$, then $G$ is called \emph{bipartite}.  If $\binom{A}{2}\cap E=\emptyset$ and $\binom{B}{2}\cap E=\binom{B}{2}$, then $G$ is called a \emph{split graph}. If $\binom{A}{2}\cap E=\binom{A}{2}$ and $\binom{B}{2}\cap E=\binom{B}{2}$, then $G$ is called \emph{cobipartite}. %$k$-colorability generalizes bipartiteness by partitioning the vertex set into $k$ independent sets instead of two.
%If there exists a partial order on the vertex set~$V$ of  $G=(V,E)$ such that $u,v\in V$ are incomparable if and only if $\{u,v\}\in E$, then $G$ is a \emph{comparability graph}.
Many graphs and graph classes can be obtained by associating geometric objects to vertices and having an edge between two vertices \iffl the corresponding objects intersect, or are disjoint, or one is contained in the other, or they overlap (i.e., they intersect but we see no containment). Considering intervals on the real line as objects, intersection yields the class of \emph{interval graphs}, disjointness gives the \emph{co-interval graphs}, containment defines the \emph{permutation graphs}, and overlap gives the \emph{circle graphs} that can be alternatively described by an intersection model of chords of a cycle. \emph{Cographs} form the smallest class of graphs containing $K_1$ and being closed under disjoint graph union and graph join. \emph{Threshold graphs} are cographs that are also split graphs. This class has many characterizations, see \cite{Gol2004e}. 

Let $G=(V,E)$ be a bipartite graph with the partition classes $A,B \subseteq V$.  Then the \emph{bipartite complement graph} of $G$ with respect to the partition $V = A \cup B$ is $(V, E')$ with $E' = \{ \{ a, b \} \mid a \in A, b \in B \} \setminus E$. \longversion{This operation was studied in~\cite{AleKLZ2020} and is denoted by $\tilde{G}$ there. }For a graph class $\cG$ of bipartite graphs, $\bico\cG$  is the class of bipartite complement graphs of graphs $G \in \cG$ with respect to any partition of $V(G)$ into two disjoint, independent sets.  The bipartite graph~$G$ is called \emph{convex} if there exists a linear ordering $<_A$ on~$A$ such that for all $b\in B$, $N(b)$ is consecutively ordered with respect to $ <_A$, i.e., there is no $a \in A\setminus N(b)$ such that there are $a_1,a_2\in N(b)$ with $a_1<_A a <_A a_2$. Also see \cite{Tuc72}.
 $G$ is called a \emph{bipartite chain graph} if there is a linear ordering $\leq_A$ on $A$ such that for each pair $a_1,a_2 \in A$, $a_1 \leq_A a_2$ implies $N(a_1) \subseteq N(a_2)$. This leads to an linear ordering $\leq_B$ on $B$ such that $b_1\leq_B b_2$ implies $N(b_1)\supseteq N(b_2)$ for all $b_1,b_2\in B$. For more information on this graph class, we refer to~\cite{Yan82}.
By definition, every bipartite chain graph is convex. 

We refrain from giving formal definitions of further notions from graph theory that we only mention in passing, like treewidth, treedepth, or degeneracy.\longversion{ Also the notion of planar graphs is quite intuitive but lengthy do define in a formal way.} We refer to textbooks on graph theory and, concerning graph classes, to~\cite{isgci}.
\shortversion{Due to space limitations, we omitted some proofs and mark this by~$(*)$.}

\paragraph*{Our Main Notions Based On \cite{FerFMS2026a}.}

We start with the definition, which allows us to compare words over an alphabet of a graph's nodes with a binary language. This morphism $h_{\ta,\tb}$ can be viewed as the composition of the projection $h_{\{u,v\}}$ and the renaming isomorphism $\iota:\{u,v\}^*\to\{0,1\}^*$ with $u\mapsto 0$ and $v\mapsto 1$
for some given nodes $u,v$. For example, if  $\ta,\tb,\tn\in V$ and $\mathtt{banana}\in V^{\ast}$, we get $h_{\ta,\tb}(\mathtt{banana})=1000$. If we have $0$-$1$-symmetric languages, it is irrelevant whether $u$ or $v$ is mapped to $0$ or $1$, respectively.

\begin{definition}
Let $V$ be an alphabet and fix $u,v\in V$ with $u\neq v$. Define $h_{u,v}:V^{\ast}\rightarrow\{0,1\}^{\ast}$ by $u\mapsto 0$, $v\mapsto 1$ and $x\mapsto \emptyword$ for all $x\in V\setminus\{u,v\}$.
\end{definition}

\noindent
Based on $h_{u,v}$, we introduce the notion of $L$-representability of graphs and of classes of graphs.
\begin{definition}

Let $L\subseteq\{0,1\}^*$ be $0$-$1$-symmetric. Let $w\in V^*$ with $V=\alphabet(w)$.
We call a graph $G=(V,E)$ \emph{$L$-represented by $w\in V^{\ast}$} iff, for all $u,v\in V$ with $u \neq v$, we have
\(\{u,v\}\in E \Leftrightarrow h_{u,v}(w)\in L. \)
A graph  $G=(V,E)$  is \emph{$L$-representable} if it can be $L$-represented by some word $w\in V^*$. 
Let $G(L,w)$ denote the graph $G$ that is $L$-represented by $w\in\Sigma^{\ast}$ and let $\cG_L=\{G(L,w)\mid w\in\Sigma^*, \Sigma=\alphabet(w)\}$ be the graph class represented by $L$.
\end{definition}

As $0$-$1$-symmetry is essential for $G(L,w)$ and $\mathcal{G}_L$ being well-defined, we will assume this condition for any binary language used for defining graphs and graph classes in the following without always mentioning it.
%We will also say that, e.g., $C_4$ is $L$-represented by some word $w\in V^*$ (with $|V|=4$), referring to some abstract graph.
For some immediate insights we refer to \cite{FerFMS2026a}.
 \longversion{For instance, t}\shortversion{T}he language $L_{\classical} \coloneqq(1\cup\emptyword)(01)^*(0\cup\emptyword) \cup (0\cup\emptyword)(10)^*(1\cup\emptyword)$ modeling classical word representability is $0$-$1$-symmetric\longversion{. This also holds for}\shortversion{ as is} its complement $\overline{L_{\classical}}$\longversion{ where $\overline{L}$ denotes $\Sigma^{\ast}\backslash L$ for some language $L\subseteq\Sigma^{\ast}$}.

\begin{example}
     \label{exa:represented-graphs-NullGraphs}
    %\todo{item now referenced later} 
    If %$L \subseteq \{\emptyword\}$, 
    $L\subseteq 0^*\cup 1^*$,
    then $\cG_L=\{N_n\mid n\in\N_{\geq 1}\}$. In order to represent $N_n=([n],\emptyset)$, take the word $w_n\coloneqq12\cdots n$\longversion{, or any other word that contains each letter from $[n]$ at least once}: for any pair of vertices $i< j$, we have $h_{i,j}(w_n)=01\notin L$.
\end{example}

\begin{example}
\label{exa:represented-graphs-CompleteUnionNullGraphs} %\todo{referenced in \autoref{tab:graphClasses} and in the proof of \autoref{cor:noSubsetInclusionImplied_1}} 
   If $L=\{01,10\}=\{0,1\}^{1\text{-uni}}$, then $\cG_L=\{K_n\cup N_m\mid n,m\in\N\}$. In order to represent $K_n\cup N_m$ as $\left([n+m],\binom{[n]}{2}\right)$, take the word $w_{n,m}\coloneqq w_n\cdot (n+1)\cdot (n+1)\cdots (n+m)(n+m)$, or any other word that contains each letter from $[n]$ exactly once and each other letter at least twice.  Clearly, for any pair of vertices $i,j\in [n]$ $i< j$, $h_{ij}(w_{n,m})=01\in L$, while for any pair $i,j\in [n+m]$, $i<j$ and $j>n$, $h_{ij}(w_{n,m})=011\notin L$ or $h_{i,j}(w_{n,m})=0011\notin L$. The argument can be generalized to any $L=\{0,1\}^{k\text{-uni}}$. 
\end{example}

\begin{example}
\label{exa:represented-graphs-CompleteBipartiteUnionNullGraphs} 
   %\todo{referenced above \autoref{prop:bipartite} and in \autoref{tab:graphClasses}} 
   If $L= \{0,1\}^{1,2\text{-uni}}$,  
    then  $\cG_L=\{K_{n,m}\cup N_\ell\mid n,m,\ell\in\N\}$. 
    \longversion{Namely, t}\shortversion{T}ake $w_{n,m,\ell}\coloneqq w_{n,m}\cdot\mbox{$(n+m+1)^3$}\linebreak[3]\cdots\linebreak[3] (n+m+\ell)^3$. Observe\shortversion{:}\longversion{ that}  $h_{i,j}(w_{n,m,\ell})=011\in L$ \iffl $1\leq i\leq n$ and $n<j\leq n+m$. Th\longversion{e argument can be generalized}\shortversion{is generalizes} to any $L=\{0,1\}^{k,\ell\text{-uni}}$ if $k\neq\ell$. 
    \end{example}

\begin{example}\label{exa:represented-graphs-classical}
%\paragraph*{Uniform and Nearly Uniform Words}
\longversion{Observe that} $\cG_{L_{\classical}}$ is the class of word-representable graphs\longversion{ (in the classical sense)}.
We can  define the finite language $L_{\text{cl-$k$-uni}}\coloneqq L_{\text{cl}}\cap \{0,1\}^{k\text{-uni}}$ and hence the class $\mathcal{G}_{L_{\text{cl-$k$-uni}}}$ of \emph{$k$-uniformly word-representable graphs} from the literature.
%Note that $L_{k\text{-uni}}\subseteq\{0,1\}^{2k}$ is a finite language.
\longversion{This is one of the motivations for us to pay special attention to graphs that are $L$-representable by finite~$L$.} Even these simple languages can possess interesting characterizations, e.g.,  by \cite[Thm. 5.1.7]{KitLoz2015}, $\cG_{\langle 0101\rangle}$ is the class of circle graphs as $\langle 0101\rangle = L_{\text{cl-2-uni}}$. 
Reinterpreting~\cite{CheKKKP2019}, the family of $k$-11-representable graphs 
can be described by $L_{k\text{-}11}$, the\longversion{ language of} binary words containing each factor $00$ and $11$ at most $k$ times, e.g.,  $L_{\text{cl}}=L_{0\text{-}11}$. Note that $\cG_{L_{2\text{-}11}}$ contains all graphs and $\cG_{L_{1\text{-}11}\cap \{0,1\}^{2\text{-uni}}}$ is the class of interval graphs~\cite{CheKKKP2019}.
\end{example}

%\begin{lemrep}\applabel{lem:clique-representation}
%Let $V$ be an alphabet and $w\in V^*$. If $0^k\shuffle 1^\ell\subseteq L\subseteq\{0,1\}^*$ for all $k,\ell\in F(w)$, then %$K_{|V|}$ is $L$-represented by~$w$.
%\end{lemrep}

%\begin{proof}
%Consider $G=(V,E)=G(L,w)$. Let $u,v\in V$ be two arbitrary letters. By definition, $|w|_u,|w|_v\in F(w)$. Hence, $h_{u,v}(w)\in 0^{|w|_u}\shuffle 1^{|w|_v}\subseteq L$, i.e., $\{u,v\}\in E$. Therefore, $G$ is a complete graph.
%\end{proof}

\paragraph*{One Example Graph - Represented By Different Words For Different Languages.}
%\todo[inline]{- example for an easy graph ($C_4$ with an extra edge) that can be represented with different finite languages [Kevin]}
\begin{figure}[bt]
\centering
\begin{minipage}[bt]{0.35\textwidth}\centering
 \scalebox{1}{\begin{tikzpicture}
    \tikzset{every node/.style={fill = white,circle,minimum size=0.05cm}}
            \node[draw] (x1) at (0,1) {$1$};
            \node[draw] (x2) at (1,1) {$2$};
            \node[draw] (x3) at (1,0) {$3$};
            \node[draw] (x4) at (0,0) {$4$};
    
            \path (x1) edge[-] (x2);
            \path (x2) edge[-] (x3);
            \path (x3) edge[-] (x4);
            \path (x4) edge[-] (x1);
            \path (x4) edge[-] (x2);
    \end{tikzpicture}} 
%\caption{Graph classes represented by certain length-uniform and nearly length-uniform languages} \label{tab:graphClasses}
\end{minipage}
\begin{minipage}[t]{0.62\textwidth}
\begin{tabular}[width=.5\textwidth]{| l |  l ||l|l|}
\hline 
$L$ & $w$ &$L$ & $w$ \\
\hline 
% $L_{\neg{1^3}}$  \scalebox{.68}{\cite[Thm. 4]{JonKPR2015}} & $11241 1234 1234$ 
$\langle 0110 \rangle$ \negthinspace  & $24131342$ &$\langle 0101, 0110 \rangle$ \negthinspace  & $12413423$ \\
%$L_{1112}$  \scalebox{.68}{\cite[Thm. 3.12]{GaeJi2020}} & $221324124311241312341234$ \\
%$L_{2\text{-}11}$  \scalebox{.68}{\cite[Thm. 5.1]{CheKKKP2019}}  & $14323412413224314231243143124312$\\
%\hline
%$L_{\cP}$ \scalebox{.68}{(Thm.~\ref{thm:palindromes-get-all})} & $4312112314$  \\
%$L_{\mathcal{C}}$ \scalebox{.68}{(Thm.~\ref{thm:copy-words-get-all})}  & $12134 12314$ \\
%$L_{\mathcal{L}}$ \scalebox{.68}{(Thm.~\ref{thm:Lyndon-words-get-all})}\negthinspace  &$111222333444 1234 1234 1124 113 234 234 2234 22 34 34 334 33 44 44 44$ \\
%$L_{\mathcal{D}}$ \scalebox{.68}{(Thm.~\ref{thm:dyck_comparability})}\negthinspace  & $3124132413241324$ \\
$\langle 0011 \rangle$ \negthinspace  & $22133144$ 
&$\langle 0011, 0110 \rangle$ \negthinspace  & $13132442$ \\
$\langle 0101 \rangle$ \negthinspace  & $13243124$ 
&$\langle 0011, 0101 \rangle$ \negthinspace  & $13312424$ \\
%$\langle 0110 \rangle$ \negthinspace  & $24131342$ &&\\
%$\langle 0101, 0110 \rangle$ \negthinspace  & $12413423$ \\
%$\langle 0011, 0110 \rangle$ \negthinspace  & $13132442$ \\
%$\langle 0011, 0101 \rangle$ \negthinspace  & $13312424$ \\
%\hline $L_{\mathcal{D}}$ & $12341234132412341234$ \\
\hline
\end{tabular}

    \end{minipage}
    \caption{Language $L$ represents $K_4-e$ (shown on the left) with the word $w$. }
    \label{fig:languages-K4-e}
\end{figure}
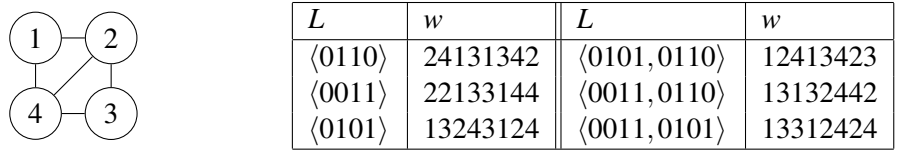

\autoref{fig:languages-K4-e} shows how the word representation for the same graph can differ when the language changes. 
%The languages $L_{\cP}$, $L_{\mathcal{C}}$, %$L_{\mathcal{L}}$,  
%and $L_{\mathcal{D}}$ refer to the symmetric hulls of palindromes, copy-languages, %Lyndon word,
%and Dyck word languages. For the first four languages, we know that $K_4-e$ (the complete graph on four nodes minus one edge) can be represented by shorter words; the given words are built with ideas from subsequent proofs. 
%Observe that especially $L_{\mathcal{L}}$ results in a very long word even when the length is (in general) in $\mathcal{O}(\vert V\vert^2)$.\todo{The comment on Lyndon words is a bit distracting; should we simply omit these here for the sake of simplification?} 
%Different from $L_{\overline{1^3}}, L_{\cP},$ and $ L_\mathcal{C}$ where we add something  for missing edges, the approach for $L_{\mathcal{L}}$ adds something if an edge exists. Since the given graph has many edges in comparison to its order, this results in a longer word. Furthermore, the approach needs a prefix of length $\vert V \vert^3$ before stating with the actual description of edges. This is more than for the other ones: sometimes such a prefix is not needed.    
Notice that in general, it is not clear, given a binary language~$L$ and some graph~$G$, how to determine an $L$-representation of~$G$ or how to decide if such a representation exists at all.
\begin{toappendix}
We are aware that it is a uniform version of the recognition problem of the graph class~$\cG_L$ to determine an $L$-representation of a graph $G$ for some given binary language $L$ and some graph $G\in\cG_L$. There are literally hundreds of papers dealing with this problem for different graph classes, indicating that we cannot expect a general answer, also because the language~$L$ might be arbitrarily complicated from a computational perspective. However, this is a natural line of research for our framework. For some language-representable graph classes, in particular those considered in the following subsection, it is known that recognition is possible in polynomial time. 
%However, it is an \NP-complete problem to recognize whether a given graph is word-representable \cite{KitLoz2015}.
Recognition is \NP-complete for word-representable graphs; see~\cite{HalKitPya2016,KitLoz2015}.
\end{toappendix}

%We will call $L\subseteq\{0,1\}^*$ \emph{length-uniform} if there is some $k\in\N$ such that $L\subseteq 0^k\shuffle 1^k$. Then, we also say that $L$ is \emph{$k$-uniform}. Similarly, $L\subseteq\{0,1\}^*$ is called \emph{nearly length-uniform} if there two different numbers $k,\ell\in \N_{\geq 1}$ such that $L\subseteq (0^k\shuffle 1^\ell)\cup (0^\ell\shuffle 1^k)$. If $k<\ell$, we also say that $L$ is \emph{$(k,\ell)$-uniform}.  Nearly by definition, $L_{k\text{-uni-}}$ is $k$-uniform. \longversion{Moreover, $\langle 001,010,011\rangle$ is $(1,2)$-uniform.}} 

\begin{figure}
\centering
\begin{minipage}[bt]{0.3\textwidth}
    $\langle 0011 \rangle$: \\
    \centering
    \begin{tikzpicture}
    \tikzset{every node/.style={fill = white,circle,minimum size=0.05cm}}
            \node[draw] (a) at (0,1.5) {$\ta$};
            \node[draw] (b) at (1,2) {$\tb$};
            \node[draw] (c) at (2,1.5) {$\tc$};
            \node[draw] (d) at (2,0.5) {$\td$};
            \node[draw] (e) at (1,0) {$\te$};
            \node[draw] (f) at (0,0.5) {$\tf$};
    
            \path (a) edge[-,color=red,line width=1.5pt] (f);
            \path (b) edge[-,color=red,line width=1.5pt] (f);
            \path (c) edge[-,color=red,line width=1.5pt] (f);
            \path (d) edge[-,color=red,line width=1.5pt] (f);
    \end{tikzpicture}
\end{minipage}
\begin{minipage}[bt]{0.3\textwidth}
    $\langle 0101 \rangle$: \\
    \centering
    \begin{tikzpicture}
    \tikzset{every node/.style={fill = white,circle,minimum size=0.05cm}}
            \node[draw] (a) at (0,1.5) {$\ta$};
            \node[draw] (b) at (1,2) {$\tb$};
            \node[draw] (c) at (2,1.5) {$\tc$};
            \node[draw] (d) at (2,0.5) {$\td$};
            \node[draw] (e) at (1,0) {$\te$};
            \node[draw] (f) at (0,0.5) {$\tf$};
    
            \path (a) edge[-,color=blue,dashed,line width=1.8pt] (b);
            \path (a) edge[-,color=blue,dashed,line width=1.8pt] (d);
            \path (a) edge[-,color=blue,dashed,line width=1.8pt] (e);
            \path (c) edge[-,color=blue,dashed,line width=1.8pt] (d);
    \end{tikzpicture}
\end{minipage}
\begin{minipage}[bt]{0.3\textwidth}
    $\langle 0110 \rangle$: \\
    \centering
    \begin{tikzpicture}
    \tikzset{every node/.style={fill = white,circle,minimum size=0.05cm}}
            \node[draw] (a) at (0,1.5) {$\ta$};
            \node[draw] (b) at (1,2) {$\tb$};
            \node[draw] (c) at (2,1.5) {$\tc$};
            \node[draw] (d) at (2,0.5) {$\td$};
            \node[draw] (e) at (1,0) {$\te$};
            \node[draw] (f) at (0,0.5) {$\tf$};
    
            \path (a) edge[-,color=cyan,dotted,line width=2pt] (c);
            \path (b) edge[-,color=cyan, dotted,line width=2pt] (c);
            \path (b) edge[-,color=cyan, dotted,line width=2pt] (d);
            \path (b) edge[-,color=cyan, dotted,line width=2pt] (e);
            \path (c) edge[-,color=cyan, dotted,line width=2pt] (e);
            \path (d) edge[-,color=cyan, dotted,line width=2pt] (e);
            \path (e) edge[-,color=cyan, dotted,line width=2pt] (f);
    \end{tikzpicture}
\end{minipage}
\begin{minipage}[bt]{0.3\textwidth}
    $\langle 0011,0101 \rangle$: \\
    \centering
    \begin{tikzpicture}
    \tikzset{every node/.style={fill = white,circle,minimum size=0.05cm}}
            \node[draw] (a) at (0,1.5) {$\ta$};
            \node[draw] (b) at (1,2) {$\tb$};
            \node[draw] (c) at (2,1.5) {$\tc$};
            \node[draw] (d) at (2,0.5) {$\td$};
            \node[draw] (e) at (1,0) {$\te$};
            \node[draw] (f) at (0,0.5) {$\tf$};

            \path (a) edge[-,color=red,line width=1.5pt] (f);
            \path (b) edge[-,color=red,line width=1.5pt] (f);
            \path (c) edge[-,color=red,line width=1.5pt] (f);
            \path (d) edge[-,color=red,line width=1.5pt] (f);
            \path (a) edge[-,color=blue,dashed,line width=1.8pt] (b);
            \path (a) edge[-,color=blue,dashed,line width=1.8pt] (d);
            \path (a) edge[-,color=blue,dashed,line width=1.8pt] (e);
            \path (c) edge[-,color=blue,dashed,line width=1.8pt] (d);
    \end{tikzpicture}
\end{minipage}
\begin{minipage}[bt]{0.3\textwidth}
    $\langle 0011,0110 \rangle$: \\
    \centering
    \begin{tikzpicture}
    \tikzset{every node/.style={fill = white,circle,minimum size=0.05cm}}
            \node[draw] (a) at (0,1.5) {$\ta$};
            \node[draw] (b) at (1,2) {$\tb$};
            \node[draw] (c) at (2,1.5) {$\tc$};
            \node[draw] (d) at (2,0.5) {$\td$};
            \node[draw] (e) at (1,0) {$\te$};
            \node[draw] (f) at (0,0.5) {$\tf$};

            \path (a) edge[-,color=red,line width=1.5pt] (f);
            \path (b) edge[-,color=red,line width=1.5pt] (f);
            \path (c) edge[-,color=red,line width=1.5pt] (f);
            \path (d) edge[-,color=red,line width=1.5pt] (f);
            \path (a) edge[-,color=cyan,dotted,line width=2pt] (c);
            \path (b) edge[-,color=cyan, dotted,line width=2pt] (c);
            \path (b) edge[-,color=cyan, dotted,line width=2pt] (d);
            \path (b) edge[-,color=cyan, dotted,line width=2pt] (e);
            \path (c) edge[-,color=cyan, dotted,line width=2pt] (e);
            \path (d) edge[-,color=cyan, dotted,line width=2pt] (e);
            \path (e) edge[-,color=cyan, dotted,line width=2pt] (f);
    \end{tikzpicture}
\end{minipage}
\begin{minipage}[bt]{0.3\textwidth}
    $\langle 0101,0110 \rangle$: \\
    \centering
    \begin{tikzpicture}
    \tikzset{every node/.style={fill = white,circle,minimum size=0.05cm}}
            \node[draw] (a) at (0,1.5) {$\ta$};
            \node[draw] (b) at (1,2) {$\tb$};
            \node[draw] (c) at (2,1.5) {$\tc$};
            \node[draw] (d) at (2,0.5) {$\td$};
            \node[draw] (e) at (1,0) {$\te$};
            \node[draw] (f) at (0,0.5) {$\tf$};

            \path (a) edge[-,color=blue,dashed,line width=1.8pt] (b);
            \path (a) edge[-,color=blue,dashed,line width=1.8pt] (d);
            \path (a) edge[-,color=blue,dashed,line width=1.8pt] (e);
            \path (c) edge[-,color=blue,dashed,line width=1.8pt] (d);
            \path (a) edge[-,color=cyan,dotted,line width=2pt] (c);
            \path (b) edge[-,color=cyan, dotted,line width=2pt] (c);
            \path (b) edge[-,color=cyan, dotted,line width=2pt] (d);
            \path (b) edge[-,color=cyan, dotted,line width=2pt] (e);
            \path (c) edge[-,color=cyan, dotted,line width=2pt] (e);
            \path (d) edge[-,color=cyan, dotted,line width=2pt] (e);
            \path (e) edge[-,color=cyan, dotted,line width=2pt] (f);
    \end{tikzpicture}
\end{minipage}
\begin{minipage}[bt]{0.3\textwidth}
    $\emptyset$: \\
    \centering
    \begin{tikzpicture}
    \tikzset{every node/.style={fill = white,circle,minimum size=0.05cm}}
            \node[draw] (a) at (0,1.5) {$\ta$};
            \node[draw] (b) at (1,2) {$\tb$};
            \node[draw] (c) at (2,1.5) {$\tc$};
            \node[draw] (d) at (2,0.5) {$\td$};
            \node[draw] (e) at (1,0) {$\te$};
            \node[draw] (f) at (0,0.5) {$\tf$};

    \end{tikzpicture}
\end{minipage}
\begin{minipage}[bt]{0.3\textwidth}
    $\langle 0011,0101,0110 \rangle$: \\
    \centering
    \begin{tikzpicture}
    \tikzset{every node/.style={fill = white,circle,minimum size=0.05cm}}
            \node[draw] (a) at (0,1.5) {$\ta$};
            \node[draw] (b) at (1,2) {$\tb$};
            \node[draw] (c) at (2,1.5) {$\tc$};
            \node[draw] (d) at (2,0.5) {$\td$};
            \node[draw] (e) at (1,0) {$\te$};
            \node[draw] (f) at (0,0.5) {$\tf$};

            \path (a) edge[-,color=red,line width=1.5pt] (f);
            \path (b) edge[-,color=red,line width=1.5pt] (f);
            \path (c) edge[-,color=red,line width=1.5pt] (f);
            \path (d) edge[-,color=red,line width=1.5pt] (f);
            \path (a) edge[-,color=blue,dashed,line width=1.8pt] (b);
            \path (a) edge[-,color=blue,dashed,line width=1.8pt] (d);
            \path (a) edge[-,color=blue,dashed,line width=1.8pt] (e);
            \path (c) edge[-,color=blue,dashed,line width=1.8pt] (d);
            \path (a) edge[-,color=cyan,dotted,line width=2pt] (c);
            \path (b) edge[-,color=cyan, dotted,line width=2pt] (c);
            \path (b) edge[-,color=cyan, dotted,line width=2pt] (d);
            \path (b) edge[-,color=cyan, dotted,line width=2pt] (e);
            \path (c) edge[-,color=cyan, dotted,line width=2pt] (e);
            \path (d) edge[-,color=cyan, dotted,line width=2pt] (e);
            \path (e) edge[-,color=cyan, dotted,line width=2pt] (f);
    \end{tikzpicture}
\end{minipage}
\caption{Graphs represented by the word $w = \mathtt{aebcdcadbffe}$ for the different sublanguages of $\langle0011,0101,0110\rangle$. Observe that $G(\langle 0011\rangle,w)=G(\langle 0101\rangle,\mathtt{fabcdfdcbaee})$.}
\label{fig:graphs-for-one-word-with-different-L}
\end{figure}
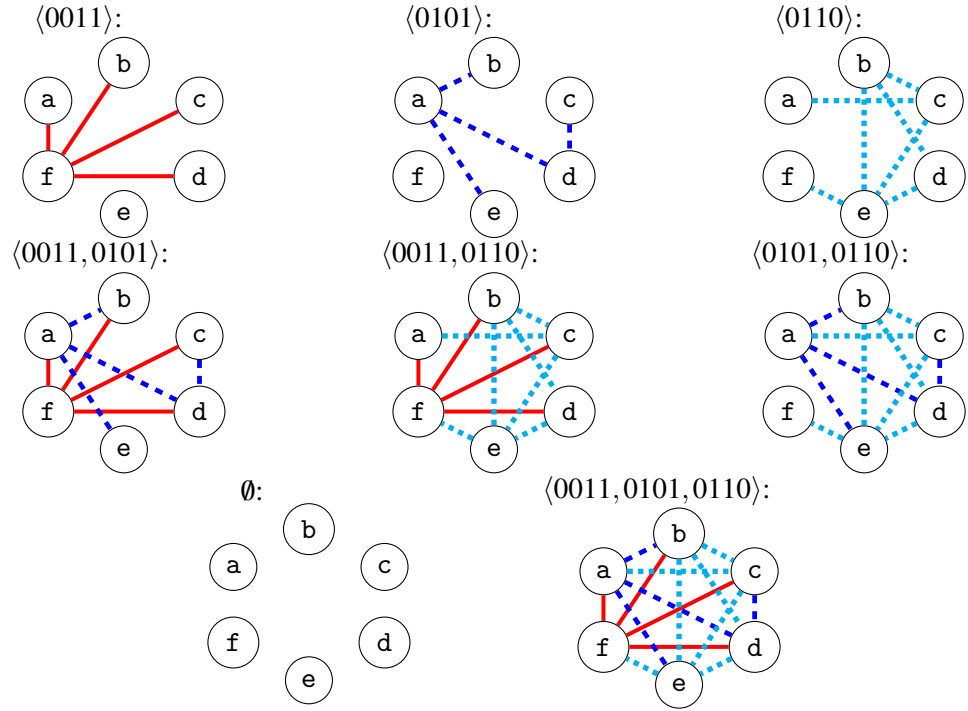

\paragraph*{One Word Represents Different Graphs With Different Languages.}
In \autoref{fig:graphs-for-one-word-with-different-L}, we fix one word $w$ and study all graphs that can be represented by this word $w$ with different $2$-uniform languages, namely all $0$-$1$-symmetric sublanguages $L$ of $\langle 0011,0101,0110 \rangle$. Note that each $2$-uniform word $v \in L$ represents a set of edges together with its complement $\widetilde{v}$. The edge sets described by two different words $v,v'$ are disjoint. By combining two words, e.g., $v = 0101$ and $v' = 0110$ we get a graph whose edge set is the union of the edge sets described by $v$ and $v'$. 

\begin{figure}
\centering

\begin{tabular}{| c | c | c |}
\hline 
$L$ & $\mathcal{G}_L$ & reference \\
\hline 
$\emptyset$ or $\{\emptyword\}$ or $\langle 0\rangle$ & $\{N_m\mid m\in\N\}$ & \autoref{exa:represented-graphs-NullGraphs} \\
$\langle 01 \rangle$ & $\{K_n\cup N_m\mid n,m\in\N\}$ & \autoref{exa:represented-graphs-CompleteUnionNullGraphs} \\
\hline
$\langle 001 \rangle$ & bipartite chain graphs & \autoref{thm:bipartite-chain}\\
$\langle 010 \rangle$ & convex graphs & \autoref{thm:convex} \\
$\langle 011 \rangle$ & bipartite chain graphs & \autoref{cor:bipartite-chain2} \\
$\langle 010,011 \rangle$ & bipartite chain graphs & \autoref{cor:bica-classes}\\
$\langle 001, 011\rangle$ & bico-convex graphs & \autoref{cor:bica-classes}\\
$\langle 010,001 \rangle$ & bipartite chain graphs & \autoref{cor:bica-classes} \\
$\langle 001,010,011 \rangle$ & $\{K_{n,m}\cup N_\ell\mid n,m,\ell\in\N\}$ & \autoref{exa:represented-graphs-CompleteBipartiteUnionNullGraphs} \\
\hline
$\langle 0011 \rangle$ & co-interval graphs & \autoref{cor:co-intervalgraphs} \\
$\langle 0101 \rangle$ & circle graphs & \cite[Thm. 5.1.7]{KitLoz2015} \\
$\langle 0110 \rangle$ & permutation graphs & \autoref{thm:permGraphs} \\
$\langle 0101,0110 \rangle$ & interval graphs & \autoref{thm:intervalgraphs} \\
$\langle 0011,0110 \rangle$ & $\{ G \cup N_m \mid G \text{ is a co-circle graph}, m \in \mathbb{N} \}$ & \autoref{thm:co-circle} \\
$\langle 0011,0101 \rangle$ & permutation graphs & \autoref{cor:co-permutation} \\
$\langle 0011,0101,0110 \rangle$ & $\{K_n\cup N_m\mid n,m\in\N\}$ & %\autoref{lem:00shuffle11} 
(*)\\
\hline
\end{tabular}
\caption{Graph classes represented by certain ``small'' binary languages.} \label{tab:graphClasses-intro}
\end{figure}

\section{General Properties}

First, we present an interplay between graph and monoid morphisms.

\begin{lemma}\label{lem:isomorphism}
    Let $G_1=(V_1,E_1)$, $G_2=(V_2,E_2)$ be graphs, $L\subseteq \{0,1\}^*$ be a language and $w=w_1\cdots w_n$ a word with $G_1=G(L,w)$. Then for each graph isomorphism $h:V_1\to V_2$ from $G_1$ to $G_2$, $G_2\simeq G(L,h(w_1)\cdots h(w_n))=G(L,h(w))$, with $h$ lifted to a monoid morphism. If $h:V_1\to V_2$ is a graph morphism, then $G(L,h(w))$ is isomorphic to a subgraph of~$G_2$.
\end{lemma}

\begin{proof}
    Let $v,u\in V$. Define $h(w)=h(w_1)\cdots h(w_n)$, i.e., interpret the graph isomorphism~$h$ now as a monoid morphism $h:V_1^*\to V_2^*$. Then $ \{h(v),h(u)\} \in E(G(L,h(w))) \iff$
    \begin{equation}
         h_{h(v),h(u)}(h(w))\in L\iff h_{v,u}(w)\in L\iff \{v,u\} \in E_1 \iff \{h(v),h(u)\} \in E_2\,. \label{eq:morphisms}
    \end{equation} Also, as $h:V_1\to V_2$ is a bijection, $|V_2|=|h(w)|=|w|=|V_1|$.
    If $h:V_1\to V_2$ is (only) a graph morphism, then the last equivalence in \autoref{eq:morphisms} will be (only) an implication.%\qed 
\end{proof}

\noindent
The following lemma and corollary are known from \cite{FerFMS2026a}: 
\begin{lemma}\applabel{lem:representability-reversal}
A graph~$G$ can be $L$-represented \iffl it can be $L^R$-represented.
\end{lemma}

\begin{corollary}\label{cor:L-symmetry}
    If $L\subseteq \{0,1\}^*$ \longversion{be a language with}\shortversion{satisfies} $L=L^R$, then $G(L,w)= G(L,w^R)$ for each word $w$.
\end{corollary}

The set operations on the sets of edges of the graphs $G(L_1,w)$ and $G(L_2,w)$ are compatible with the set operations of $L_1$ and $L_2$ in the following sense:

\begin{lemrep} \applabel{lem:edgesets}
Let  $L_1, L_2 \subseteq \{ 0, 1 \}^\ast$ be %0-1-symmetric 
languages. Let $w\in V^*$, with 
$V=\alphabet(w)$. \\
%The set operations on the sets of edges of the graphs $G(L_1,w)$ and $G(L_2,w)$ are compatible with the set operations of $L_1$ and $L_2$ in the sense that the following statements hold: 
1. Let $\Box$ be a binary set operation \longversion{in the sense}\shortversion{such} that for a set $X$, subsets $A, B \subseteq X$ and a binary Boolean operation $\Diamond:\{0,1\}^2\to\{0,1\}$, we set $A\mathbin{\Box} B=\{x\in X\mid x\in A \mathbin{\Diamond} x\in B\}$. Then, $E(G(L_1 \mathbin{\Box} L_2,w)) = E(G(L_1,w)) \mathbin{\Box} E(G(L_2,w))$.\\
2. $E(G(\overline{L_2},w)) = \binom{V}{2} \setminus E(G(L_2,w))$ and \longversion{$G(\overline{L_2},w)$ is the complement graph of $G(L_2,w)$, i.e., }$G(\overline{L_2},w)=\overline{G(L_2,w)}$.\\
3. If $L_1 \subseteq L_2$, then $E(G(L_1,w)) \subseteq E(G(L_2,w))$\longversion{, i.e., $G(L_1,w)$ is a subgraph of $G(L_2,w)$}.
%\item $E(G(L_1 \cap L_2,w)) = E(G(L_1,w)) \cap E(G(L_2,w))$.
%\item $E(G(L_1 \setminus L_2,w)) = E(G(L_1,w)) \setminus E(G(L_2,w))$.
%\item $E(G(L_1 \cup L_2,w)) = E(G(L_1,w)) \cup E(G(L_2,w))$.
\end{lemrep}

\shortversion{As an example for the taste of the algebraic proof we argue for item~1: $\{ u, v \} \in E(G(L_1 \mathbin{\Box} L_2,w))  \iff h_{u,v}(w) \in L_1 \mathbin{\Box} L_2  \iff h_{u,v}(w) \in L_1 \mathbin{\Diamond} h_{u,v}(w) \in L_2  \iff \{ u, v \} \in E(G(L_1,w)) \mathbin{\Diamond} \{ u, v \} \in E(G(L_2,w))  \iff \{ u, v \} \in E(G(L_1,w)) \mathbin{\Box} E(G(L_2,w))$. \qed}

\begin{proof}
\begin{enumerate}
\item %Let $\Diamond:\{0,1\}^2\to\{0,1\}$ be the binary Boolean operation corresponding to the binary\todoReviewer{“to the binary” -> “to a binary”} set operation $\Box:2^X\times 2^X\to2^X$ in the sense that $A\mathbin{\Box} B=\{x\in X\mid x\in A \mathbin{\Diamond} x\in B\}$.  
For $v \in V$ and $u \in V \setminus \{ v \}$, the following holds: 
\begin{align*}
\{ u, v \} \in E(G(L_1 \mathbin{\Box} L_2,w)) & \iff h_{u,v}(w) \in L_1 \mathbin{\Box} L_2 \\
                                    & \iff h_{u,v}(w) \in L_1 \mathbin{\Diamond} h_{u,v}(w) \in L_2 \\
                                    & \iff \{ u, v \} \in E(G(L_1,w)) \mathbin{\Diamond} \{ u, v \} \in E(G(L_2,w)) \\
                                    & \iff \{ u, v \} \in E(G(L_1,w)) \mathbin{\Box} E(G(L_2,w))
\end{align*}
%For $v \in V$ and $u \in V \setminus \{ v \}$, the following holds: 
%\begin{align*}
%\{ u, v \} \in E(G(L_1 \cap L_2,w)) & \iff h_{u,v}(w) \in L_1 \cap L_2 \\
%                                    & \iff h_{u,v}(w) \in L_1 \wedge h_{u,v}(w) \in L_2 \\
%                                    & \iff \{ u, v \} \in E(G(L_1,w)) \wedge \{ u, v \} \in E(G(L_2,w)) \\
%                                    & \iff \{ u, v \} \in E(G(L_1,w)) \cap E(G(L_2,w))
%\end{align*}
%\item Analogously to (2). %For $v \in V$ and $u \in V \setminus \{ v \}$, the following holds: 
\item The statement follows from (1): $L_1 = \{ 0, 1 \}^\ast$ and $\Diamond: (b_1, b_2) \mapsto b_1 (1 - b_2)$.
%\item The statement follows from (2) and (4) by De Morgan's laws.

%\item For $v \in V$ and $u \in V \setminus \{ v \}$, the following withholds: 
%$$\{ u, v \} \in E(G(L_1,w)) \Rightarrow h_{u,v}(w) \in L_1 \Rightarrow h_{u,v}(w) \in L_2 \Rightarrow \{ u, v \} \in E(G(L_2,w))$$
\item The statement follows from (1) with $\Diamond: (b_1,b_2) \mapsto b_1 b_2$ as, for $L_1,L_2\subseteq \{0,1\}^*$, $L_1\subseteq L_2$ \iffl $L_1\cap L_2=L_1$.
\qed
\end{enumerate}\renewcommand{\qed}{}
\end{proof}

This lemma has interesting consequences when we turn our attention to\longversion{wards} graph classes instead of single languages and graphs. \longversion{The next statement is a reinterpretation of}\shortversion{We restate} the second statement of \autoref{lem:edgesets}.

\begin{proposition}\label{prop:compl}
For $L \subseteq \{ 0, 1 \}^* $, $\mathcal{G}_{\overline{L}}=\{\overline{G}\mid G\in \mathcal{G}_L\}$. 
%is exactly the set of complement graphs of graphs from~$\mathcal{G}_L$. 
\end{proposition} 

\noindent
The following lemma about projective morphisms is helpful in various settings.

\begin{lemma} \label{lem:inducedsubgraph} 
Let $L\subseteq \{0,1\}^*$ be a language, $V$ an alphabet and $w \in V^*$, with $V=\alph(w)$, and $G = G(L,w)$. Let $A\subseteq V$. 
%Let $G[A] = (A, E \cap \binom{A}{2})$ be the subgraph of $G$ induced by $A$. 
Then, $G[A] = G(L,h_A(w))$. 
\end{lemma}

\begin{proof} 
Let  $G = (V,E) = G(L,w)$. 
Let $G' = G(L,h_A(w))$. Obviously, $V(G') = A = V(G[A])$. For every $u \in A$ and $v \in A \setminus \{ u \}$, the following holds: 
\shortversion{$\{ u, v \} \in E(G[A]) \iff{}$
%\begin{align*}
\[  \{ u, v \} \in E 
                       \iff{}  h_{u,v}(w) \in L 
                       \iff{}  h_{u,v}(h_A(w)) \in L 
                       \iff{}  \{ u, v \} \in E(G') 
                       \]
%\end{align*}
}
\longversion{
\begin{eqnarray*}
    \{ u, v \} \in E(G[A]) &\iff&\{ u, v \} \in E \\
     &\iff&  h_{u,v}(w) \in L \\&\iff&  h_{u,v}(h_A(w)) \in L 
                       \\&\iff&  \{ u, v \} \in E(G') 
\end{eqnarray*}

}
Therefore, $E(G')=E \cap \binom{A}{2}$.%\qed
\end{proof}

\noindent
This lemma immediately entails the following \longversion{important and interesting }property of any language-representable graph class.
%that is definable in our framework.
%, also mentioned in \autoref{thm:hereditary-intro}. 
\longversion{Notice how we again switch our attention from languages to graph classes.}

\begin{theorem}\label{thm:hereditary}
For any $L\subseteq\{0,1\}^*$, the graph class $\cG_L$ is hereditary\longversion{, i.e., closed under induced subgraphs}.
\end{theorem}

This property was known before for word-representable graphs \cite[Prop. 3.0.8]{KitLoz2015} and is known for  the graph classes that we identify in the following as being $L$-representable for some language~$L$. However, we have a simple algebraic argument that works for all of these graph classes.

Notice that \autoref{thm:hereditary} can be interpreted in a way saying that our graph classes are closed under removing vertices. A natural related question would be if they are also closed under adding vertices. \longversion{Clearly, we cannot expect a positive answer in the sense that we ignore the possible connections of the new vertex with the old ones: this would mean that we always describe all graphs. However, under some conditions, certain types of adding vertices are possible, as we explore in the following.
The next proposition exploits `holes' in the set of frequentnesses.}
\begin{prprep}\applabel{propos:closure-adding-isolates} 
For any $L\subseteq\{0,1\}^*$, if $\N_{\geq 1}\setminus \mathrm{freq}(L)\neq\emptyset$, then $\cG_L$ is closed under adding isolated vertices and  $\cG_{\overline{L}}$ is closed under adding universal vertices.
\end{prprep}

\begin{proof}
Consider some $L\subseteq\{0,1\}^*$ such that there is some `missing frequentness', say, $m\in \N_{\geq 1}\setminus \mathrm{freq}(L)$.
Let $G=(V,E)\in\cG_L$. Hence, there is some $w\in V^*$ with $G=G(L,w)$. Now, let $\ta\notin V$. Define $w'=w\cdot \ta^m$. Then, $G'=G(L,w')$ is isomorphic to $G\cup N_1$.
Consider $H\in \cG_{\overline{L}}$. By \autoref{prop:compl}, $\overline{H}\in \cG_L$. By the previous reasoning, $\overline{H}\cup N_1\in \cG_L$. Now, observe that $\overline{\overline{H}\cup N_1}=H\nabla N_1$ by De Morgan's Law. Inductively, the claims follow.
\end{proof}

\begin{proposition} \label{prop:bipartite}
%For $k, \ell \in \mathbb{N}$, i  
Let $k,\ell \in \mathbb{N}$ such that $0 < k < \ell$.
If $L\subseteq \{0,1\}^{k,\ell\text{-uni}}$, then all graphs in~$\cG_L$ are bipartite.
\end{proposition}

\begin{proof} For $k,\ell \in \mathbb{N}$ such that $0 < k < \ell$,  let $L\subseteq \{0,1\}^{k,\ell\text{-uni}}$. For $G \in \mathcal{G}_L$, choose $w \in V(G)^*$ such that $G(L,w) = G$. Let $V_i = \{ v \in V(G) \mid |w|_v = i\}$ for $i\in\{k, \ell \}$. Let, $u,v \in V_k$.
Hence, $ |h_{u,v}(w)|_0 = |w|_u = k = |w|_v = |h_{u,v}(w)|_1$. As $k\neq \ell$, $h_{\{u,v\}}(w) \notin \{0,1\}^{k,\ell\text{-uni}}\supseteq L$. Hence, $V_k$ is independent;  analogously,  $V_\ell$ is independent. This proves that $G$ is bipartite.
\longversion{For $\ell \in \mathbb{N}$ such that $0 < \ell$, the $(0,\ell)$-uniform, $0$-$1$-symmetric languages are $\{ 0^\ell, 1^\ell \}$ and $\emptyset$. Hence $\mathcal{G}_L$ is the class of null graphs. Those are bipartite.}
%\qed
\end{proof}

%Combining this with the argument of \autoref{lem:clique-representation}, one obtains, generalizing \autoref{exa:represented-graphs-CompleteBipartiteUnionNullGraphs}:
%\begin{lemma}\label{lem:biclique-representation}
%For $0 < k < \ell$, if $L=0^k \shuffle 1^\ell$, then $G(L,w)=K_{n,m}\cup N_r$ for some $n,m,r\in\N$.
%\end{lemma}

\autoref{prop:bipartite} is also one of the building blocks of the following type of decomposition theorem. It means that, given a graph $G=G(L,w)$, we can first classify $\alph(w)$ by its frequencies in~$w$ and then decompose the edge set of~$G$ as either belonging to the graph induced by the letters of frequency~$k$, or as belonging to the bipartite graph between any two different frequencies $k,\ell\in F(w)$. This helps understand the structure of such graphs and allows to build them from more basic graphs by their (not necessarily disjoint) union. 
To formulate it, we introduce some more formalisms. For $w\in V^*$ and $k,\ell\in F(w)$, set $V_{k,\ell}(w)\coloneqq\{x\in V\mid |w|_x\in\{k,\ell\}\}$. We write $g_{w,k,\ell}$ instead of $h_{V_{k,\ell}(w)}$.
%the morphism $g_{w,k,\ell}:V^*\to V^*$ defined for $x\in V$ by $x\mapsto x$ if $|w|_x\in\{k,\ell\}$ and $x\mapsto\emptyword$, otherwise. 
For $L \in \{ 0,1 \}^\ast$ and a word $w$, we define  $E_{k,\ell}(L,w)\coloneqq\{\{u,v\}\in E(G)\mid |w|_u=k\land |w|_v=\ell\}$, giving $G_{k,\ell}(L,w)=(V_{k,\ell}(w),E_{k,\ell}(L,w))$ which is $G[V_{k,\ell}(w)]$ if $k=\ell$ and a bipartite subgraph of $G[V_{k,\ell}(w)]$ if $k\neq\ell$. %, i.e.,  %For $k = \ell$, we define $V_k(w) = V_{k,k}(w)$. If the language $L$ and the word $w$, or the word $w$ respectively, are clear from the context, we might simply write $E_{k,\ell}$, $V_{k,\ell}$ and~$V_k$\longversion{ for the respective sets}.

\begin{theorem}\label{thm:graph-decomposition}
 Let $L\subseteq \{0,1\}^*$. 
 Let $G=(V,E)\in \cG_L$ be $L$-represented by $w\in V^*$.
 Let $P_{L,w}\coloneqq\{(k,\ell)\in\N^2\mid 0 < k\leq \ell\land k,\ell\in F(w)\}$.
 %\exists x,y\in V: |w|_x = k \wedge |w|_y = \ell \wedge h_{x,y}(w) \in L\}$. 
 Then,  $E=\bigcup_{(k,\ell)\in P_{L,w}}E_{k,\ell}(L,w)\,.$
 Moreover,  $G_{k,\ell}(L,w)\in \cG_{\{0,1\}^{k,\ell\text{-uni}}\cap L}$.
\end{theorem}

\begin{proof}
%For all $(k,\ell) \in P_{L,w}$, $G[V_{k,\ell}(w)] = G(L,g_{w,k,\ell}(w))= (V_{k,\ell}(w),E_{k,\ell}(L,w)) $, because $g_{w,k,\ell} = h_{V_{k,\ell}(w)}$. This follows from \autoref{lem:inducedsubgraph}. Hence, $G_{k,\ell}$ is well-defined and $E_{k,\ell}(L,w) =E \cap \binom{V_{k,\ell}(w)}{2}$. 
To prove $E=\bigcup_{(k,\ell)\in P_{L,w}}E_{k,\ell}(L,w)$, we only need to show the inclusion $\subseteq$. Let $\{ u, v \} \in E$. W.l.o.g., assume $|w|_u \leq |w|_v$. In this case, $(|w|_u,|w|_v) \in P_{L,w}$ and $u, v \in V_{|w|_u,|w|_v}(w)$. Hence, $\{ u, v \} \in E_{|w|_u,|w|_v}(L,w) \subseteq \bigcup_{(k,\ell)\in P_{L,w}}E_{k,\ell}(L,w)$. 

For all $(k,\ell) \in P_{L,w}$ and for all $u, v \in V_{k,\ell}(w)$ with $u \neq v$, $h_{u,v}(g_{w,k,\ell}(w)) \in\{0,1\}^{k,\ell\text{-uni}}$. Hence, \shortversion{$h_{u,v}(g_{w,k,\ell}(w)) \in L$ \iffl $h_{u,v}(g_{w,k,\ell}(w)) \in \{0,1\}^{k,\ell\text{-uni}} \cap L$.}\longversion{$$h_{u,v}(g_{w,k,\ell}(w)) \in L\iff h_{u,v}(g_{w,k,\ell}(w)) \in \{0,1\}^{k,\ell\text{-uni}} \cap L\,.$$} This implies $G_{k,\ell}(L,w) = G(\{0,1\}^{k,\ell\text{-uni}}\cap L,\shortversion{\linebreak[4]} g_{w,k,\ell}(w))$. Thus, $G_{k,\ell}(L,w)\in \cG_{\{0,1\}^{k,\ell\text{-uni}}\cap L}$. %\qed
\end{proof}

\begin{example}\label{exa:threshold} 
Consider $L=\langle 001,01\rangle$. By \autoref{thm:graph-decomposition} and \autoref{exa:represented-graphs-CompleteUnionNullGraphs}, $\cG_L$ is a family of split graphs. More precisely, if $G=(V,E)\in \cG_L$, represented by some $w\in V^*$, then $E=E_{1,1}\cup E_{1,2}$. In particular, all vertices that occur at least thrice in $w$ have degree zero. Hence, the projective morphism that only preserves vertices that occur once in~$w$, let $V_1\subseteq V$ be this set, allows us to apply \autoref{lem:inducedsubgraph}, concluding that $G[V_1]=G(L,h_{V_1}(w))$ is a complete graph, while (using a projective morphism preserving vertices from $V\setminus V_1$) $G[V\setminus V_1]=G(L,h_{V\setminus V_1}(w))$ is a null graph, so that~$G$ is a split graph. %Conversely, if $G=(V,E)$ is a split graph, then $V$ can be decomposed into $C$ and~$I$ such that $G[C]$ is a complete graph and $I$ is an independent set. 
%More details on this graph class can be found in \autoref{subsec:atmost-two-occurrences}.
We will see in \autoref{thm:bipartite-chain} that $\langle 001\rangle$ describes the bipartite chain graphs. 
As it is known that $G$ is a bipartite chain graph \iffl $G'$, obtained from $G$ by turning any side of the bipartition into a clique, is a threshold graph, $\cG_L$ indeed characterizes the threshold graphs.
\end{example} 

\autoref{thm:graph-decomposition} motivates to focus on (nearly) length-uniform languages, as graphs described with the help\longversion{ of other,} more complex languages can be decomposed in the way described by the preceding theorem.
Also, as (nearly) length-uniform languages are finite, this motivates the study of finite~$L$ and their graph classes.

\paragraph*{Dealing With Syntactic Sugar.}

Lastly, we are formulating a \emph{trash theorem}. \begin{toappendix}
This trash theorem will be useful when trying to apply \autoref{prop:compl} in the context of the study, say, of 2-uniform languages and the corresponding graph classes, because then, we rather need a complementation relative to all 2-uniform words instead of relative to all possible words. To understand the difference between these two types of complements of a binary language~$L$, we now study their difference formally, which gives the set $T_L$ as defined in the following theorem, which can be viewed as a collection of a sort of trash words with only minor influence on the graph classes.
In short, the difference between these versions of complementation is often just some syntactic sugar that can be ignored.
\end{toappendix}

%\begin{comment}
%\todo{The following Theorem is false! Add some sentences in the text saying this.}
%\begin{theorem}
%Let $L\subseteq \{0,1\}^*$ be such that  $\mathcal{G}_{L}$ is closed under adding universal vertices. Let $\mathrm{lengths}(L)=\{n\in\mathbb{N}\mid \exists w\in L: |w|=n\}$. Let $\hat L=L\cup \bigcup_{n\in \mathbb{N}\setminus \mathrm{lengths}(L)}\{0,1\}^n$. Then,     $\mathcal{G}_{L}=\mathcal{G}_{\hat L}$.
%\end{theorem}
%\end{comment}

%For a word $w$ we call the number of occurrences of a letter $v \in \alph(w)$ the \emph{frequentness} of $v$ in $w$. 

\begin{thmrep}[Trash Theorem] \applabel{trash-theorem} 
Let $L\subseteq \{0,1\}^*$. % be a $0$-$1$-symmetric language. Let $\mathrm{freq}(L)=\{n\in\mathbb{N}\mid \exists w\in L: |w|_0 = n\}$. 
Let $\hat L \coloneqq L \cup T_L$, where we define the trash as $T_L \coloneqq  \{ w \in \{ 0, 1 \}^* \mid |w|_0 \notin \mathrm{freq}(L) \vee |w|_1 \notin \mathrm{freq}(L) \} $.  Then, $L$ and $T_L$ are disjoint. 
\begin{enumerate}
\item If $\mathcal{G}_{\hat L}$ is closed under adding isolated vertices, then $\mathcal{G}_{L} = \mathcal{G}_{\hat L}$ and for every $G \in \cG_L$ there exists a $w \in V(G)^\ast$ such that $|w|_v \in \mathrm{freq}(L)$ for every $v \in V(G)$ and $G(L,w) = G(\hat L, w)$. 
\item If $\mathcal{G}_{L}$ is closed under adding universal vertices, then $\mathcal{G}_{\hat L} \subseteq \mathcal{G}_{L}$. 
\end{enumerate}
\end{thmrep}
\begin{proof} 
We assume $L$ to be $0$-$1$-symmetric. Hence for every $w\in L$, $|w|_0 \in \mathrm{freq}(L)$ and $|w|_1 \in \mathrm{freq}(L)$. This shows that $L$ and $T_L$ are disjoint. We can assume, w.l.o.g., that $\mathrm{freq}(L) \neq \mathbb{N}_{\geq 1}$. (Otherwise, $L = \hat L$ and the proof is trivial.) 

\begin{enumerate}
\item We will show the inclusion $\mathcal{G}_{L} \subseteq \mathcal{G}_{\hat L}$ first. Consider a graph $G = (V,E) \in \mathcal{G}_{L}$. Hence, there exists a word $w \in V^*$ such that $G = G(L, w)$. Partition $V$ into the sets  \shortversion{$V_1 = \{ v \in V \mid |w|_v \notin \mathrm{freq}(L) \}$ and $V_2 = \{ v \in V \mid |w|_v \in \mathrm{freq}(L) \}$.}\longversion{$$V_1 = \{ v \in V \mid |w|_v \notin \mathrm{freq}(L) \}\text{ and }V_2 = \{ v \in V \mid |w|_v \in \mathrm{freq}(L) \}\,.$$} $V_1$ is a set of isolated vertices, also confer \autoref{propos:closure-adding-isolates}. $G[V_2] = G(L,h_{V_2}(w))$, because of \autoref{lem:inducedsubgraph}. By definition of $V_2$, $G(L,h_{V_2}(w)) = G(\hat L,h_{V_2}(w))$. Hence, $G[V_2] \in \mathcal{G}_{\hat L}$. Therefore, $G = G[V_2] \cup (V_1, \emptyset) \in \mathcal{G}_{\hat L}$, because $\mathcal{G}_{\hat L}$ is closed closed under adding isolated vertices. 

We will show $\cG_{\hat L} \subseteq \mathcal{G}_{L}$ and the rest of the statement next. For the sake of contradiction, suppose a graph $G \in \mathcal{G}_{\hat L}$ exists such that, for every $w \in V(G)^\ast$ with $G = G(\hat L, w)$, a vertex $t \in V(G)$ exists such that $|w|_t \notin \mathrm{freq}(L)$. We call this statement $(*)$. We want to add a new isolated vertex $v \notin V(G)$ to~$G$. Let $G' = G \cup (\{v \}, \emptyset)$. Since $G' \in \cG_{\hat L}$, a word $w'$ exists such that $G(\hat L, w') = G'$. If $|w'|_v \notin \mathrm{freq}(L)$, $\{ u, v \} \in E(G')$ for every $u \in V(G)$. This is a contradiction, because $V(G)$ is not empty and $v$ is supposed to be isolated in~$G'$. Hence, $|w'|_v \in \mathrm{freq}(L)$. If $|w'|_u \notin \mathrm{freq}(L)$ for a vertex $u \in V(G)$, $\{ u, v \} \in E(G')$. This is a contradiction. Hence, for every $x \in V(G) \cup \{ v \}$, $|w'|_x \in \mathrm{freq}(L)$. Therefore, $G(\hat L, w') = G(L, w')$ and $G(L, h_{V(G)}(w')) = G$. This is a contradiction to $(*)$. Hence, the negation of $(*)$ is true, i.e., for every $G \in G_{\hat L}$, a word $w$ with $G = G(\hat L, w)$ exists such that, for every $t \in V(G)$, $|w|_t \in \mathrm{freq}(L)$. Hence, $G = G(L,w)$ and $G \in \cG_L$.  

\item Consider a graph $G = (V,E) \in \mathcal{G}_{\hat L}$. Hence, there exists a word $w \in V^*$ such that $G = G(\hat L, w)$. Partition $V$ into the sets  $V_1 = \{ v \in V \mid |w|_v \notin \mathrm{freq}(L) \}$ and $V_2 = \{ v \in V \mid |w|_v \in \mathrm{freq}(L) \}$. For each $u\in V$ and $v\in V \setminus \{ u \}$, the following holds: 
\begin{align*}
\{ u, v \} \in E \iff{} & h_{u,v}(w) \in \hat L \\
                 \iff{} & h_{u,v}(w) \in T_L \vee (h_{u,v}(w) \in L \wedge h_{u,v}(w) \notin T_L) \\
                 \iff{} & |w|_u \notin \mathrm{freq}(L) \vee |w|_v \notin \mathrm{freq}(L) \\
                                 & \vee (h_{u,v}(w) \in L \wedge |w|_u \in \mathrm{freq}(L) \wedge |w|_v \in \mathrm{freq}(L) )\\
                 \iff{} &  u \in V_1 \vee v \in V_1 \vee (h_{u,v}(w) \in L \wedge u \in V_2 \wedge v \in V_2)
\end{align*} 
This shows that $V_1$ is a set of universal vertices. For $u,v \in V_2$, $h_{u,v}(w) = h_{u,v}(h_{V_2}(w))$ and hence $G = G(L,h_{V_2}(w)) \nabla V_1$. Since $G(L,h_{V_2}(w)) \in \mathcal{G}_L$ and $\mathcal{G}_L$ is closed under adding universal vertices, $G \in \mathcal{G}_L$. 
\qed
\end{enumerate}\renewcommand{\qed}{}
\end{proof}

\shortversion{\noindent}
Further, observe that $T_L=T_{\overline{L}}$ \iffl $L\cap \{0,1\}^{\ell\text{-uni}}\neq \{0,1\}^{\ell\text{-uni}}$ for all $\ell\in \mathrm{freq}(L)$.
We sketch one application of the trash theorem in the following.

%\todohf{Another possible trash theorem is \cite[Theorem 7]{KitPya2008}. Or more positively formulated, this is a kind of normal form that might also hold in more general terms.}

\begin{toappendix}
\begin{remark}\label{rem:adding-isolates}
The trivial way to add isolated vertices to a word $w$ describing a graph $G(L,w)$ for a finite language $L$ is by choosing the number of occurrences of the new vertex $v$ outside of $\mathrm{freq}(L)$; see \autoref{propos:closure-adding-isolates}. However, this technique cannot be used to guarantee that $\cG_{\hat L}$ is closed under adding isolated vertices. Namely, $\cG_{\hat L}$ is closed under adding isolated vertices \iffl one can add an isolated vertex $v$ in such a way that the number of occurrences of $v$ is in $\mathrm{freq}(L)$. 
\end{remark}
\end{toappendix}

\begin{correp}\applabel{cor:threshold}
$\cG_{\widehat{\langle 01, 001 \rangle}}$ and $\cG_{\langle 010, 011 \rangle \cup \{0,1\}^{2\text{-uni}}}$ are the threshold graphs. 
\end{correp}
\begin{proof}
For $L = \langle 01, 001 \rangle$, $T_L=\{ w \in \{ 0, 1 \}^* \mid |w|_0 \notin \{1,2\} \vee |w|_1 \notin \{1,2\} \}$.
Let $G = (V,E)\in \cG_{\hat L}$ and $w \in V^*$ such that $G = G(\hat L, w)$. If $v \in V$ such that $|w|_v > 2$, define $w' = v \cdot h_{V \setminus \{ v \}}(w)$. Clearly $G = G(L,w')$. By induction, we can show that a word $x \in V^*$ exists such that $|x|_v \in \{ 1, 2\}$ for every $v \in V$ and $G=G(L,x)$. Let $u$ be a vertex such that $u \notin V$. $G(\hat L,x \cdot uu) = G \cup (\{ u \}, \emptyset)$. Hence, $\cG_{\hat L}$ is closed under adding isolated vertices. By \autoref{trash-theorem}, we get $\cG_{\hat L} = \cG_{L}$. 
Clearly, ${\langle 010, 011 \rangle \cup \{0,1\}^{2\text{-uni}}} ={\overline{\hat L}}$. By \autoref{prop:compl}, $\cG_{\langle 010, 011 \rangle \cup \{0,1\}^{2\text{-uni}}}$ is the class of threshold graphs \longversion{because the threshold graphs}\shortversion{as they} are closed under graph complement and because $\cG_{L}$ are the threshold graphs as we know from \autoref{exa:threshold}. %\qed 
\end{proof}
We \longversion{conclude this section with}\shortversion{are next} giving an example how reasoning about operations on words, languages and graphs \longversion{may yield}\shortversion{gives} inclusion\longversion{ result}s for graph classes. The \longversion{proof }idea \longversion{of the following theorem }is to define two specific operations on words of even length \longversion{(that are special cases of (literal) shuffle operations; see \cite{Ber87}) so that they }\shortversion{for} graph union and join.

\begin{theorem}\label{thm:cographs}%\shortversion{$(*)$}
Let $L\subseteq\{0,1\}^*$ contain either both $1001$ and $0110$ or both $1010$ and $0101$ (but not all four words). Then, ${\cal G}_L$ contains all cographs.
\end{theorem}

\begin{proof}
Clearly, a $K_1$ can be described by the word $vv$.
This word, and all words~$w$ that we will construct in the following in order to describe a cograph, has the following property: If $w\in V^*$ describes a graph of order~$n=|V|$, then $|w|=2n$ and $w$ can be \emph{evenly decomposed} into $w=w_1w_2$ such that $|w_1|=|w_2|=n$ and every vertex from~$V$ occurs exactly once in $w_1$ and exactly once in $w_2$.\footnote{In the context of classical word-representability, this property is also known as \emph{permutationally representable}.}

Hence, let $G_1=(V_1,E_1)$ and $G_2=(V_2,E_2)$, with $V_1\cap V_2=\emptyset$, be two cographs that, by induction hypothesis, can be described by words $w$ and $u$ that can be evenly decomposed as $w=w_1w_2$ and $u=u_1u_2$. Define $x=w_1u_1w_2u_2$ and $y=w_1u_1u_2w_2$. Set $G(x)=(V,E_x)$ and $G(y)=(V,E_y)$ with $V=V_1\cup V_2$. Clearly, $E_1\cup E_2\subseteq E_x\cap E_y$. 

Now, assume that $L$ contains both $1001$ and $0110$ but neither $1010$ nor  $0101$.
Then, $G(x)=G_1\cup G_2$ and $G(y)=G_1 \nabla G_2$. 
Conversely, if $L$ contains both $1010$ and $0101$ but neither $1001$ nor $0110$, then
$G(x)=G_1\nabla G_2$ and $G(y)=G_1 \cup G_2$. 

In both cases, by induction we see that all cographs are contained in the graph class~${\cal G}_L$.%\qed
\end{proof}

\begin{corollary}
The classical word-representable graphs contain all cographs. \longversion{More precisely, as cographs can be represented by 2-uniform words, they have representation number at most~2. Moreover, there are permutationally representable by 2-uniform words.}
\end{corollary}

\begin{toappendix}
One can find  properties similar to \autoref{thm:cographs} for $L$ based on words of length $kn$ for any $k\geq 2$ that also guarantee that all cographs can be represented, by giving other decomposition conditions for these words\longversion{, but these conditions grow increasingly complicated without giving too many further insights. Therefore, we refrain from stating them explicitly.} %But this way, one can easily see that}
%\shortversion{. Similarly,} any class  $\mathcal{G}_{L_{k\text{-uni}}}$ (for $k\geq 2$) contains all cographs.
\end{toappendix}

\section{Graph Classes From Small Finite Languages}\label{sec:specialGraphClassesFiniteLanguages}
\label{sec:special-results}

\autoref{tab:graphClasses-intro} shows the results of this section for $1$-, $(1,2)$- and $2$-uniform languages. 
%They have been also explicitly summarized and highlighted as \autoref{thm:small-languages-intro} in the introduction. 
\longversion{Notice that mostly we can characterize graph classes known from the graph theory literature, but sometimes (namely, when this class is not closed under adding isolated vertices), we have to make unions of null graphs explicit due to \autoref{propos:closure-adding-isolates}, which applies to all finite languages and their related graph classes. In a sense, this could be then seen as the smallest superclass of graphs that is hereditary.}

\longversion{\subsection{Interval Models With Respect To $(1,2)$- and $2$-Uniform Languages}}
\shortversion{\paragraph*{Interval Models With Respect To $(1,2)$- and $2$-Uniform Languages.}}
\longversion{In this subsection, w}\shortversion{W}e will show that there is a uniform way to interpret graph classes defined by $(1,2)$- and $2$-uniform languages~$L$ geometrically. This explains why we encounter many well-known graph classes as such classes~$\cG_L$.
Interestingly, not all of these geometric models have been described before to the best of our knowledge.

\begin{definition} \label{def:interval-model}
Let $L$ be a $(1,2)$-uniform or a $2$-uniform language such that $\emptyset \subsetneq L \subsetneq \{0,1\}^{1,2\text{-uni}}$, or $\emptyset \subsetneq L \subsetneq \{0,1\}^{2\text{-uni}}$, respectively. Let $\mathcal{I} = \{ I_v \}_{v \in V}$ be a family of closed intervals on a linearly ordered set $\mathcal{D}$ such that for each interval $I \in \mathcal{I}$ an endpoint of $I$ is never an endpoint of another interval.\longversion{ We allows intervals where left and right endpoint coincide. We say that $\mathcal{I}$ is}\shortversion{ Call $\mathcal{I}$}\emph{ an interval model of $G$ with respect to the language $L$} if, for each $v \in V$ and for each $u \in V \setminus \{ v \}$, $\{u,v\} \in E$ \iffl the pair $(I_u,I_v)$ matches the pattern of a word contained in $L$. \longversion{Here, t}\shortversion{T}he pair $([x_\alpha, x_\beta],[y_\alpha, y_\beta])$ \emph{matches the pattern of \longversion{the word} $w \in L$} if 
\\[.3ex]
\begin{minipage}{.48\textwidth}
\begin{itemize}
\item $ x_\alpha < x_\beta < y_\alpha < y_\beta$ for $w = 0011$,
\item $x_\alpha <y_\alpha < x_\beta < y_\beta$ for $w = 0101$,
\item $x_\alpha < y_\alpha < y_\beta < x_\beta$ for $w = 0110$,
\item $x_\alpha < x_\beta < y_\alpha = y_\beta$ for $w = 001$, \end{itemize} 
\end{minipage}\begin{minipage}{.52\textwidth}
\begin{itemize}
\item $x_\alpha < y_\alpha = y_\beta < x_\beta$ for $w = 010$,
\item $x_\alpha = x_\beta < y_\alpha < y_\beta$ for $w = 011$, or
\item $([y_\alpha, y_\beta],[x_\alpha, x_\beta])$ matches the pattern of the word $\widetilde{w}$. 
\end{itemize} \end{minipage}
\end{definition}
\shortversion{\noindent}
The different kinds of interval models \longversion{yielded by}\shortversion{of} \autoref{def:interval-model} are illustrated in \autoref{fig:interval-model}. 
\begin{figure}[tb] \centering
\begin{tikzpicture}
\newcommand{\dy}{0.11}
\newcommand{\drawinterval}[3]{\draw (#1,#3+\dy) -- (#1,#3-\dy) -- (#1,#3) -- (#2,#3) -- (#2,#3+\dy) -- (#2,#3-\dy);}

\draw (0,1) node[anchor=west,align=left]{0 0 1 1};
\drawinterval{0.1}{0.6}{0.7}
\drawinterval{0.7}{1.2}{0.7}
\draw (0,0) node[anchor=west,align=left]{disjoint\\intervals};

\draw (4,1) node[anchor=west,align=left]{0 1 0 1};
\drawinterval{4.1}{4.9}{0.7}
\drawinterval{4.4}{5.2}{0.5}
\draw (4,0) node[anchor=west,align=left]{intersecting intervals,\\no inclusion (overlap)};

\draw (8,1) node[anchor=west,align=left]{0 1 1 0};
\drawinterval{8.1}{9.2}{0.7}
\drawinterval{8.4}{8.9}{0.5}
\draw (8,0) node[anchor=west,align=left]{one interval\\contains the other};

\draw (0,-1) node[anchor=west,align=left]{0 0 1};
\drawinterval{0.1}{0.6}{-1.3}
\filldraw[black] (0.8,-1.3) circle (2pt);
\draw (0,-2) node[anchor=west,align=left]{point to the right\\of the interval};

\draw (4,-1) node[anchor=west,align=left]{0 1 0};
\drawinterval{4.1}{4.9}{-1.3}
\filldraw[black] (4.5,-1.5) circle (2pt);
\draw (4,-2) node[anchor=west,align=left]{point contained\\in the interval};

\draw (8,-1) node[anchor=west,align=left]{0 1 1};
\drawinterval{8.4}{8.9}{-1.3}
\filldraw[black] (8.2,-1.3) circle (2pt);
\draw (8,-2) node[anchor=west,align=left]{point to the left\\of the interval};
\end{tikzpicture} 
\caption{Interval pairs matching words from $(1,2)$- and $2$-uniform languages (with points $p$ considered intervals of the form $[p,p]$)} \label{fig:interval-model}
\end{figure}
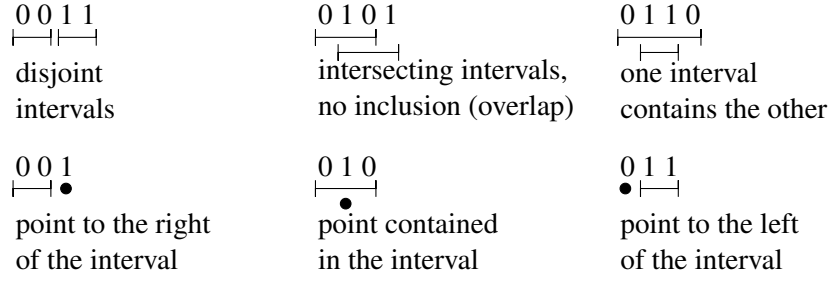

\begin{example} Consider\longversion{ the language} $L = \langle 010 \rangle$. For each graph described by an interval model with respect to $L$, the vertices can be partitioned into the vertices described by intervals that are points and intervals that are not points. There are only edges between those two vertex sets. \longversion{A vertex described by a point is adjacent to a vertex described by an interval that is not a point \iffl the point is contained in the interval.}\shortversion{A ``point vertex'' is adjacent to an ``interval vertex''  \iffl the point is contained in the interval.} This is illustrated in \autoref{fig:interval-model-convex}.
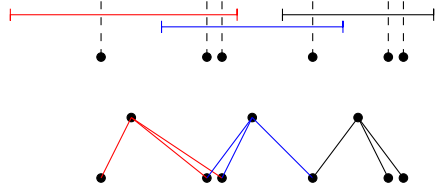
\begin{figure} \begin{center}
\begin{tikzpicture}[scale=0.8]
\newcommand{\dy}{0.1}
\newcommand{\drawinterval}[4]{\draw[#1] (#2,#4+\dy) -- (#2,#4-\dy) -- (#2,#4) -- (#3,#4) -- (#3,#4+\dy) -- (#3,#4-\dy);}

\filldraw (1.50,0) circle (2pt); \draw[dashed] (1.50,0) -- (1.50,1);
\filldraw (3.25,0) circle (2pt); \draw[dashed] (3.25,0) -- (3.25,1);
\filldraw (3.50,0) circle (2pt); \draw[dashed] (3.50,0) -- (3.50,1);
\filldraw (5.00,0) circle (2pt); \draw[dashed] (5.00,0) -- (5.00,1);
\filldraw (6.25,0) circle (2pt); \draw[dashed] (6.25,0) -- (6.25,1);
\filldraw (6.50,0) circle (2pt); \draw[dashed] (6.50,0) -- (6.50,1);

\drawinterval{red}  {0.0}{3.75}{0.7}
\drawinterval{blue} {2.5}{5.50}{0.5}
\drawinterval{black}{4.5}{7.00}{0.7}

\filldraw (1.50,-2) circle (2pt);
\filldraw (3.25,-2) circle (2pt);
\filldraw (3.50,-2) circle (2pt);
\filldraw (5.00,-2) circle (2pt);
\filldraw (6.25,-2) circle (2pt);
\filldraw (6.50,-2) circle (2pt);

\filldraw (2.00,-1) circle (2pt);
\draw[red] (2.00,-1) -- (1.50,-2);
\draw[red] (2.00,-1) -- (3.25,-2);
\draw[red] (2.00,-1) -- (3.50,-2);

\filldraw (4.00,-1) circle (2pt);
\draw[blue] (4.00,-1) -- (3.25,-2);
\draw[blue] (4.00,-1) -- (3.50,-2);
\draw[blue] (4.00,-1) -- (5.00,-2);

\filldraw (5.75,-1) circle (2pt);
\draw[black] (5.75,-1) -- (5.00,-2);
\draw[black] (5.75,-1) -- (6.25,-2);
\draw[black] (5.75,-1) -- (6.50,-2);
\end{tikzpicture}
\end{center}
\caption{An interval model with respect to\longversion{ the language} $\langle 010 \rangle$ and the graph described\longversion{ by the model}.} \label{fig:interval-model-convex}
\end{figure}
\end{example}

The following lemma is not only interesting for understanding when we can build interval models\longversion{, see \autoref{thm:interval-model},}\shortversion{ (\autoref{thm:interval-model})} but it is also crucial for \longversion{the graph class property of }adding isolate\longversion{d vertice}s \longversion{as explained in \autoref{rem:adding-isolates}}\shortversion{by descibing them with letters occurring thrice}.  Also confer \autoref{trash-theorem} and \autoref{prop:2-uniform-plus-vertices}. 

\begin{lemrep} \applabel{lem:uniformword}
1. For every $L \subsetneq \{0,1\}^{1,2\text{-uni}} $ and for every $G \in \cG_L$, there  \longversion{exists a word}\shortversion{is a} $w \in V(G)^\ast$ such that $G = G(L,w)$ and $|w|_v \in \{ 1, 2 \}$ for every $v \in V(G)$. \\
2. For every $L \subsetneq \{0,1\}^{2\text{-uni}} $ such that $0110 \notin L$ or $0011 \notin L$ and for every $G \in \cG_L$, there \longversion{exists a word}\shortversion{is a} $w \in V(G)^\ast$ \longversion{such that}\shortversion{with} $G = G(L,w)$ and $|w|_v = 2$ for every $v \in V(G)$. 
\end{lemrep}

\begin{proof}
We only prove the first statement formally. The second statement can be proven analogously. 

Let $G \in \cG_L$. Hence there exists a word $w' \in V(G)^\ast$ such that $G = G(L,w')$. Define $V' = \{ v \in V(G) \mid |w'|_v \notin \{ 1, 2\} \}$ and let $V' = \{ v_1, \dots, v_k \}$ for $k = |V'|$. All the vertices in $V'$ are isolated. Define $$
w=
\begin{cases}
v_1^2 \cdots v_k^2 \cdot h_{V(G) \setminus V'}(w'), & \text{if } 001 \notin L\\
v_1 \cdots v_k \cdot h_{V(G) \setminus V'}(w') \cdot v_1 \cdots v_k, & \text{if } 010 \notin L\, \land\, 001 \in L\\
h_{V(G) \setminus V'}(w') \cdot v_1^2 \cdots v_k^2, & \text{if } 010 \in L\, \land\,  001 \in L.
\end{cases}
$$
If $001 \in L$ and $010 \in L$, then $011 \notin L$, because $L \neq \{0,1\}^{1,2\text{-uni}}=\langle\{001,010,100\}\rangle$ and $L$ is $0$-$1$-symmetric. Hence $G(L,w) = G(L,w') = G$ and $|w|_v \in \{ 1, 2 \}$ for every $v \in V(G)$.

Concerning the second statement, notice that the case distinction necessary to define a 2-uniform word that is equivalent to a given word is suggested by the cases given as conditions in the statement of the lemma. We need the conditions ``$0110 \notin L$ or $0011 \notin L$'' as the possible case $0101\notin L$ is not helpful for encoding isolates. %\qed
\end{proof}

%\begin{corollary}
%For every $L \subsetneq (0 \shuffle 1^2) \cup (1 \shuffle 0^2) $ or  $L \subsetneq 00 \shuffle 11 $ such that $0110 \notin L$ or $0011 \notin L$, the graph class $\cG_L$ is closed under adding isolated vertices. 
%\end{corollary}

%\todohfInPlace{We can claim that our geometric model is simpler than the one proposed in~\cite{BriLozSta2016}.}

The necessity of the conditions  $L \neq\{0,1\}^{1,2\text{-uni}} $ and  $L \neq \{0,1\}^{2\text{-uni}} $ in \autoref{lem:uniformword} is easy to see if one observes that 
$\{0,1\}^{1,2\text{-uni}}=\langle  001,010,011\rangle $ and\longversion{ that}
$\{0,1\}^{2\text{-uni}}=\langle 0011, 0101,0110\rangle\,.$
Consulting \autoref{tab:graphClasses-intro}, one sees that the\shortversion{se}\longversion{ corresponding} graph classes contain the complete (bipartite) graphs and all null graphs. For $\langle  001,010,\shortversion{\linebreak[4]}011\rangle$ and $\langle 0011, 0101,0110\rangle $, non-trivial null graphs can be only represented by letters whose frequentness is different from 1 and 2 (in the case of $\langle  001,010,011\rangle$) or different from~2 (in the case of $\langle 0011, \shortversion{\linebreak[4]}0101,0110\rangle $). The necessity of the condition ``$0110 \notin L$ or $0011 \notin L$'' is explained in the following\longversion{, less trivial example}.

\begin{prprep}\applabel{exm:verylongexample}
There are graphs $G\in \cG_L$ for $L \coloneqq \langle 0011,0110\rangle$ for which there exists no 2-uniform word $w$ such that $G=G(L,w)$, for example, $C_5\cup K_1$. 
\end{prprep}

%\begin{example}\label{exm:verylongexample}
%    Interestingly, there are graphs $G\in \cG_L$ for $L \coloneqq \langle 0011,0110\rangle$ for which there exists no 2-uniform word $w$ such that $G=G(L,w)$. 
\begin{proof}
One such example is $G=(V \coloneqq \{a,b,c,d,e,f\},E)$ with $$E \coloneqq\{\{a,b\},\{b,c\},\{c,d\},\{d,e\},\{e,a\}\})\,.$$ $G$ is isomorphic to $C_5\cup K_1$. %\todoscsInPlace{Kevin, I think we should say that $G$ is isomorphic to $C_5$.}\todokmInPlace{Silas, $G$ is isomorphic to $C_5\cup K_1$.} 
    First of all, $G=G(L,eacdabdebcf)$. Assume, for the sake of contradiction, that there is a 2-uniform word $w$ with $G=G(L,w)$. %For $A \subseteq V$, let $h_A: V \rightarrow A$ be the homomorphism with $h_A(v_1) = v_1$ for each $v_1 \in A$ and $h_A(v_2) = \emptyword$ for each $v_2 \in V \setminus A$.

    As $I_1\coloneqq \{a,c,f\}$ is a independent set, $$h_{I_1}(w)\in \{acfacf, afcafc, cafcaf, cfacfa, facfac, fcafca\}\,.$$ The function $h: V \to V$ with $h(x)=x$ for $x\in \{ b, f\}$, $h(c)=a$, $h(a)=c$, $h(d)=e$ and $h(e)=d$, is a graph automorphism, i.e., an isomorphism from $G$ to~$G$. Further, $L$ is $0$-$1$-symmetric. By \autoref{lem:isomorphism} and \autoref{cor:L-symmetry}, we only need to consider $h_{I_1}(w) \in \{acfacf, afcafc\}$. This gives two main cases with a number of subcases.
        
    \noindent \textbf{Case 1.} $h_{I_1}(w) = acfacf$. Consider $I_2 \coloneqq \{a,d,f\}$.
    
    \noindent \textbf{Case 1.a} $h_{I_2}(w) = dafdaf$.  Then $h_{I_1 \cup I_2}(w)= dacfdacf$ and $h_{d,c}(w)=0101$.
     
    \noindent \textbf{Case 1.b} $h_{I_2}(w) = adfadf$.  Then $h_{I_1 \cup I_2}(w) \in \{adcfacdf, acdfadcf\}$ as otherwise $h_{d,c}(w)=0101$. The function $h': V \to V$ with $h'(x)=x$ for $x\in \{ a, f\}$, $h'(c)=d$, $h'(d)=c$, $h'(b)=e$ and $h'(e)=b$, is a graph automorphism. Since $h(adcfacdf) = acdfadcf$, we only need to consider $h_{I_1 \cup I_2}(w) = adcfacdf$. Consider $I_3 \coloneqq \{ b,d,f\}$
     
    \noindent \textbf{Case 1.b.I} $h_{I_3}(w) = bdfbdf$. $h_{I_1 \cup I_3}(w) = badcfacbdf$ , as otherwise $h_{a,b}(w)=0101$ or $h_{b,c}(w)=0101$. Consider $I_4 \coloneqq \{e,b,f\}$.
     
    \noindent \textbf{Case 1.b.I.i} $h_{I_4}(w) = ebfebf$. Thus, $h_{e,d}(w) = 0101$.
     
    \noindent \textbf{Case 1.b.I.ii} $h_{I_4}(w) = befbef$. Hence, $h_{a,e}(w) = 0101$ or $h_{c,e}(w) = 0110$. 
     
    \noindent \textbf{Case 1.b.I.iii} $h_{I_4}(w) = bfebfe$. Then, $h_{d,e}(w) = 0101$. 
     
    \noindent \textbf{Case 1.b.II} $h_{I_3}(w) = dbfdbf$. This implies $h_{a,b}(w)=0101$.
     
    \noindent \textbf{Case 1.b.III} $h_{I_3}(w) = dfbdfb$. $h_{I_1 \cup I_3}(w) = adcfacbdfb$ , as otherwise $h_{a,b}(w)=0101$ or $h_{b,c}(w)=0101$. Consider $I_4 \coloneqq \{e,b,f\}$.

    \noindent \textbf{Case 1.b.III.i} $h_{I_4}(w) = efbefb$. Thus, $h_{a,e}(w)=0101$ or $h_{c,e}(w)=0110$.   
     
    \noindent \textbf{Case 1.b.III.ii} $h_{I_4}(w) = febfeb$. Therefore,  $h_{d,e}(w)=0101$.  
     
    \noindent \textbf{Case 1.b.III.iii} $h_{I_4}(w) = fbefbe$. This implies $h_{c,e}(w)=0011$. 
     
    \noindent \textbf{Case 1.c} $h_{I_2}(w) = afdafd$. So, $h_{c,d}(w)=0101$.
    
    \noindent \textbf{Case 2} $h_{I_1}(w) = afcafc$. 
    
    \noindent \textbf{Case 2.a} $h_{I_2}(w) = dafdaf$. Hence, $dafdcafc$. Otherwise, $h_{d,c}(w)=0101$. 
    
    \noindent \textbf{Case 2.a.I} $h_{I_3}(w) = bdfbdf$. Thus $h_{I_1\cup I_3}=bdafbdcafc$ and $h_{a,b}(w) = 1010$. 
    
    \noindent \textbf{Case 2.a.II} $h_{I_3}(w) = dbfdbf$. Therefore,  $h_{I_1\cup I_3}=dabfdbcafc$. Otherwise, $h_{b,c}(w) = 0101 $ or $ h_{a,b}(w) = 1010$. 
    
    \noindent \textbf{Case 2.a.II.i} $h_{I_4}(w) = ebfebf$. This implies,  $h_{c,e}(w)=1100$.

    \noindent \textbf{Case 2.a.II.ii} $h_{I_4}(w) = befbef$. So  $h_{d,e}(w)=0101$.
    
    \noindent \textbf{Case 2.a.II.iii} $h_{I_4}(w) = bfebfe$. Hence  $h_{a,e}(w)=0101$.
    
    \noindent \textbf{Case 2.a.III} $h_{I_3}(w) = dfbdfb$. Therefore,   $h_{a,b}(w) = 0101$. 
    
    \noindent \textbf{Case 2.b} $h_{I_2}(w) = adfadf$. Hence, $adfcadfc$ and $h_{d,c}(w)=0101$.
        
    \noindent \textbf{Case 2.c} $h_{I_2}(w) = afdafd$. Then there are two cases for $h_{I_1 \cup I_2}(w)$: $afdcafcd$ and $afcdafdc$. We can use again $h'$. Therefore, we only have to consider $afdcafcd$.

    \noindent \textbf{Case 2.c.I} $h_{I_3}(w) = bfdbfd$. Hence,  $h_{I_1\cup I_3}(w) = abfdbcafcd$. Otherwise, $h_{b,c}(w) = 0101 $ or $ h_{a,b}(w) = 1010$. 

    \noindent \textbf{Case 2.c.I.i} $h_{I_4}(w) = ebfebf$. Thus,  $h_{c,e}(w) = 1100$.

    \noindent \textbf{Case 2.c.I.ii} $h_{I_4}(w) = befbef$. Thus,  $h_{d,e}(w) = 1010$.

    \noindent \textbf{Case 2.c.I.iii} $h_{I_4}(w) = bfebfe$. Thus,  $h_{a,e}(w) = 0101$.

    \noindent \textbf{Case 2.c.II} $h_{I_3}(w) = fbdfbd$. Then $h_{a,b}(w)=0101$.
    
    \noindent \textbf{Case 2.c.III} $h_{I_3}(w) = fdbfdb$. Then $h_{a,b}(w)=0101$ or $h_{b,c}(w)=1010$.

    Since we discussed all of the cases, there is no 2-uniform word $w$ with $G=G(L,w)$.%\qed
\end{proof} 

\begin{theorem} \label{thm:interval-model}
Let $L$ be a $(1,2)$-uniform language such that $L \subsetneq \{0,1\}^{1,2\text{-uni}}$ or let $L$ be a $2$-uniform language such that $0110 \notin L$ or $0011 \notin L$. Then, $G \in \mathcal{G}_L$ \iffl an interval model of $G$ with respect to $L$ exists.
\end{theorem}
\begin{proof}
We are going to show both directions of the logical equivalence. First, let $G =(V,E) \in \mathcal{G}_L$. \longversion{According to}\shortversion{By} \autoref{lem:uniformword}, \longversion{the following holds.}\shortversion{we know:} If $L$ is $(1,2)$-uniform, then a word $w \in V^\ast$ exists \longversion{such that}\shortversion{with} $G = G(L,w)$ and $|w|_v \in \{ 1, 2 \}$ for every $v \in V$. If $L$ is $2$-uniform, then a word $w \in V^\ast$ exists such that $G = G(L,w)$ and $|w|_v = 2$ for every $v \in V$. For every $v \in V$, let $v_\alpha$ be the index of the first and $v_\beta$ be the index of the last occurrence of $v$ in $w$ and define $I_v\coloneqq[v_\alpha, v_\beta]$\shortversion{, where $I_v=[v_\alpha, v_\alpha]$ if $|w|_v=1$}\longversion{. Note that $I_v$ has the form $[v_\alpha, v_\alpha]$ if $v$ only occurs once in $w$}. Let $\mathcal{I} = \{ I_v \}_{v \in V} $. For every $v \in V$ and $u \in V \setminus \{ u \}$, $(I_v,I_u)$ matches the pattern of $h_{v,u}(w)$ by construction. Hence, $\mathcal{I}$ is an interval model of $G$ \longversion{with respect to}\shortversion{and}~$L$. 

\longversion{For the other direction}\shortversion{Conversely}, let $G =(V,E)$ and $\mathcal{I} = \{ I_v \}_{v \in V}$ be an interval model of $G$ with respect to $L$. Let $i_1 < i_2 < \dots < i_m$ be  the endpoints of the intervals \longversion{contained in}\shortversion{from}~$\mathcal{I}$, such that for an interval \longversion{$I$ of the form }$I=[i_j, i_j]$ the endpoint $i_j$ only occurs once in the list. For each $j \in [m]$ and\longversion{ each} $v \in V$, choose $w_j = v$ if $i_j$ is an endpoint of the interval~$I_v$. Define $w \coloneqq w_1 \cdots w_m$. For every $v \in V$ and $u \in V \setminus \{ u \}$, $(I_v,I_u)$ matches the pattern of $h_{v,u}(w)$ by construction.
Hence, $G(L,w) = G$ and $G \in \mathcal{G}_L$.\qed
\renewcommand{\qed}{}
\end{proof}

\begin{remark} 
For intervals $[x_\alpha, x_\beta]$ with $x_\alpha < x_\beta$, the intervals do not have to be closed \longversion{because of our model assumption}\shortversion{as we assume} that endpoints of different intervals are distinct. 
\end{remark}

\longversion{\subsection{$(1,2)$-Uniform Languages}}\shortversion{\paragraph*{$(1,2)$-Uniform Languages.}} 
\longversion{In the following}\shortversion{Next}, we consider all $(1,2)$-uniform  languages~$L\subseteq\{0,1\}^*$\longversion{ where each letter occurs either once or twice.}. This leads us to consider the following cases:
%\begin{enumerate}
%    \item 
$L_1=\langle 001\rangle$,
%    \item 
$L_2=\langle 010\rangle$,
%    \item 
$L_3=\langle 011\rangle$,
%    \item 
$L_4=\langle 001,010\rangle$,
%    \item 
$L_5=\langle 001,011\rangle$,
%    \item 
$L_6=\langle 010,011\rangle$, and
%    \item 
\longversion{$L_7=\langle 001,010,011\rangle=\{0,1\}^{1,2\text{-uni}}$}\shortversion{$L_7=\{0,1\}^{1,2\text{-uni}}$}.
%\end{enumerate
By \autoref{cor:L-symmetry}, $\cG_{L_1}=\cG_{L_3}$ as 
\longversion{$L_1^R=\{001,110\}^R=\{100,011\}=L_3$}\shortversion{$L_1^R=\{001,110\}^R=L_3$}. $\cG_{L_4}=\cG_{L_6}$. This leaves us with five languages to consider\longversion{ in this subsection}. All \longversion{graph }classes contain only bipartite graphs due to \autoref{prop:bipartite}\longversion{, but we will make this bipartition explicit sometimes in the following}. Disregarding one trivial case, we will see bipartite chain graphs, convex graphs and their respective complements\longversion{ characterized by these five languages. This is also summarized in}\shortversion{; see} \autoref{tab:graphClasses-intro}.
\autoref{thm:interval-model} implies the following; also confer \autoref{fig:interval-model-convex}: 
\begin{corollary} \label{thm:convex} 
    A graph $G=(V,E)$ is %$\langle 1011\rangle$
    $\langle 010\rangle$-representable \iffl $G$ is a convex graph.
\end{corollary} 
%\todoscsInPlace{Everyone, is it obvious how this follows from \autoref{thm:interval-model} or do we need to say something about it?}
%\todoscsInPlace{add a reference to \autoref{fig:interval-model-convex} here}

\begin{theorem} \label{thm:bipartite-chain}
%    A graph $G=(V,E)$ is %$\langle 0111\rangle$
%    $\langle 001\rangle$-representable \iffl $G$ is a bipartite chain graph.
$\mathcal{G}_{\langle 001 \rangle}$ is the class of bipartite chain graphs.
\end{theorem} 
%\todoscs{Reviewer 1: The proof of Theorem 5.8 is too long; an acceptable proof would be pointing out that $2K_2$ is not representable (easy to see), so all graphs from the class are bipartite chain (the class is equivalent to $2K_2$-free bipartite graphs); conversely, the word $a_1a_1b_1a_2a_2b_2...a_na_nb_n$ very clearly produces the prime bipartite chain graph on $2n$ vertices (and every bipartite chain graph is induced in such a prime one).}
%\todoscs{Reviewer 1: In general, many of the following theorems do not make use of standard results from the literature and of the symmetries between the classes studied.}
\begin{proof}
%To simplify the notation, define $L \coloneqq  \langle 001\rangle$.
Let $G=(V,E)\in \mathcal{G}_{\langle 001 \rangle}$. Th\longversion{erefore}\shortversion{us}, there exists a word $w$ (of length $\ell$) \longversion{such that}\shortversion{with} $G=G( \langle 001\rangle,w)$. By \autoref{lem:uniformword}, \longversion{$w$ contains only letters  which appear exactly once or twice}\shortversion{$w\in V^{1,2\text{-uni}}$}. Set $A$ to the set of vertices which appear exactly once in~$w$ and set $B$ to the set the vertices which appears exactly twice, which define the classes of a bipartite graph. %It is possible to observe  that $2K_2$ is not representable, so all graphs from the class are bipartite chain (the class is equivalent to $2K_2$-free bipartite graphs). Alternatively, a proof building on the definition of bipartite chain graphs is not much longer:
For $v\in V$, define $\max_{v}\coloneqq \max\{i\in [\ell]\mid w_i=v\}$. For all $v,u\in A$, define the ordering relation $v \leq_A u$ \iffl $\max_{v}\leq \max_u$. Let $v,u\in A$ with $v \leq_A u$ and $x\in B\cap N(v)$. Hence, $h_{x,v}(w)=001$. As $\max_{v} \leq \max_{u}$, $h_{x,u}(w)=001$. Thus, $v \leq_A u$ implies $N(v)\subseteq N(u)$.\longversion{

} Conversely, consider the word $w_n=a_1a_1b_1a_2a_2b_2\cdots a_na_nb_n$. In $G( \langle 001\rangle,w_n)$, $N(a_i)=\{b_i,b_{i+1},\dots,b_n\}$. Hence, this graph is the prime bipartite chain graph on $2n$ vertices, called $H_{2,n}$ in~\cite{LozRud2007}, where it is shown that every bipartite chain graph is an induced subgraph of some prime bipartite chain graph. By \autoref{thm:hereditary}, this proves that every bipartite chain graph is\longversion{ contained} in $\mathcal{G}_{\langle 001 \rangle}$.%\qed
\end{proof}
%\\
%comAlso, we can try to generalize the following theorem. BTW: If my claim is correct, then this also implies that all bipartite chain graphs are convex, which follows only a bit indirectly on graphclasses.org by looking at bp chain included in probe bipartite chain included in convex.}
%\todokm{To "BTW": By definition, every bipartite chain graph is convex. You can even use the same ordering.}
\noindent
By \autoref{lem:representability-reversal}, we get the next result.
\begin{corollary}\label{cor:bipartite-chain2}
$\mathcal{G}_{\langle 011 \rangle}$ is the class of bipartite chain graphs.
\end{corollary}

\begin{thmrep} \applabel{thm:bico-complement}
For every\longversion{ language}~$L$ such that $\emptyset \neq L \subsetneq \{0,1\}^{1,2\text{-uni}} $, we find that $\cG_{\overline{L} \cap \{0,1\}^{1,2\text{-uni}}} = \bico\cG_L$. 
\end{thmrep}

\begin{proof}
\begin{enumerate}
\item Let $G \in \cG_{\overline{L} \cap \{0,1\}^{1,2\text{-uni}}}$. According to \autoref{lem:uniformword}, there exists a word $w \in V(G)$ such that $G = G(\overline{L} \cap \{0,1\}^{1,2\text{-uni}},w)$ and $|w|_v \in \{ 1, 2 \}$ for every $v \in V(G)$. Clearly, $\overline{L} \cap \{0,1\}^{1,2\text{-uni}} = \{0,1\}^{1,2\text{-uni}} \setminus L$. Hence according to \autoref{lem:edgesets}, $E(G) = E(G(\{0,1\}^{1,2\text{-uni}}, w)) \setminus E(G(L, w))$. Let $V_1 = \{ v \in V(G) \mid |w|_v = 1 \}$ and $V_2 = \{ v \in V(G) \mid |w|_v = 2 \}$. Then $V_1$ and $V_2$ are two disjoint sets that partition $V(G)$ and both sets are independent in $G(L, w)$. Clearly, $E(G(\{0,1\}^{1,2\text{-uni}}, w)) = \{ \{ a, b \} \mid a \in V_1, b \in V_2 \}$. Hence, $G(L, w))$ is the bipartite complement graph of $G$ with respect to the partition of $V(G)$ into $V_1$ and $V_2$. 
\item For the other direction, consider a graph $G \in \bico\cG_L$. Hence, $G$ is the bipartite complement of a graph $G' \in \cG_L$. According to \autoref{lem:uniformword}, there exists a word $w \in V(G')^\ast$ such that $G' = G(L,w)$ and $|w|_v \in \{ 1, 2 \}$ for every $v \in V(G')$. Let $V_1 = \{ v \in V(G') \mid |w|_v = 1 \}$ and $V_2 = \{ v \in V(G') \mid |w|_v = 2 \}$. 
Let $A, B$ be two independent sets such hat $A \cup B$ is a partition of $V(G')$. Let $V(G') = K_1 \cup \dots \cup K_k$ be the partition of $V(G')$ into the connected components of $G'$. Assume there exists an $i \in [k]$ such that $K_i \cap V_1 \cap A \neq \emptyset$ and $K_i \cap V_2 \cap A \neq \emptyset$. Choose $v_1 \in K_i \cap V_1 \cap A$ and $v_2 \in K_i \cap V_2 \cap A$. Because $v_1$ and $v_2$ are in the same connected component, there must be a path $(p_0 = v_1, p_1, \dots, p_m = v_2)$. Because of $L \subseteq \{0,1\}^{1,2\text{-uni}}$, $E(G') \subseteq \{ \{ a, b \} \mid a \in V_1, b \in V_2 \}$. Hence for every $1 \leq j \leq m$, $j$ is even \iffl $p_j \in V_1$. Because of the choice of $A$ and $B$ for every $1 \leq j \leq m$, it also holds that $j$ is even \iffl $p_j \in A$. But this means that $m$ must be odd (because $p_m \in V_2$) and $j$ must be even (because $p_m \in A$). This is a contradiction. Analogously, we show that no $i \in [k]$ exists such that $K_i \cap V_1 \cap B \neq \emptyset$ and $K_i \cap V_2 \cap B \neq \emptyset$. Therefore, we can conclude the following: 
For each partition that could have been used to form the bipartite complement of the graph $G' \in \cG_L$, $E(G) = \{ \{ a, b \} \mid a \in V_1, b \in V_2 \} \setminus E(G')$. 

We use \autoref{lem:edgesets} again to prove \shortversion{$\{ \{ a, b \} \mid a \in V_1, b \in V_2 \} \setminus E(G') = E(G(\overline{L} \cap \{0,1\}^{1,2\text{-uni}}, w))$.}\longversion{$$\{ \{ a, b \} \mid a \in V_1, b \in V_2 \} \setminus E(G') = E(G(\overline{L} \cap \{0,1\}^{1,2\text{-uni}}, w))\,.$$} Hence, $G = G(\overline{L} \cap \{0,1\}^{1,2\text{-uni}},w)$. 
\qed
\end{enumerate}\renewcommand{\qed}{}
\end{proof}

\begin{correp} \applabel{cor:bica-classes}
1. $\cG_{\langle 001, 011\rangle}$ is the class of bico-convex graphs.\\
2. $\cG_{\langle 010,011 \rangle}$ and $\cG_{\langle 010,001 \rangle}$ are equal to the class of bipartite chain graphs.
\end{correp}
\begin{proof}
\begin{enumerate}
\item The statement follows from \autoref{thm:convex} and \autoref{thm:bico-complement}.
\item The statement follows from \autoref{lem:representability-reversal}, \autoref{thm:bipartite-chain}, \autoref{thm:bico-complement} and the fact that the class of bipartite chain graphs is closed under bipartite complements, which can be obtained by defining a linear ordering $\leq_{A'}$ on $A$\longversion{ such that $a_1 \leq_{A'} a_2 \Leftrightarrow a_1 \geq_A a_2$, for each pair $a_1,a_2\in A$}. 
\longversion{Alternatively, this can be also understood by the characterization of bipartite chain graphs as bisplit graphs.}
\qed
\end{enumerate}\renewcommand{\qed}{}
\end{proof}

\begin{remark}
With \autoref{thm:interval-model} we can define interval models for\longversion{ the following classes}:\\
1. bipartite chain graph (interval models with respect to the languages $\langle 001 \rangle$, $\langle 011 \rangle$, $\langle 010,011 \rangle$ and $\langle 010,001 \rangle$), as well as\\
2. bico-convex graphs (interval model with respect to the language $\langle 001, 011\rangle$).\\
We are unaware of any previous\longversion{ly published geometric} intersection models for these graph classes.
\end{remark}

\begin{prprep}\applabel{prp:convex-nonclosure-bibcomplement}
    Different from bipartite chain graphs, convex graphs are not closed under bipartite complement.
\end{prprep}
\begin{proof}
    To see this, we consider a concrete example in the following. Let $G=(A\cup B, E)$ with $A\coloneqq [6]$, $B\coloneqq \{a,b,c,d,e\}$ and 
    $$E\coloneqq \{\,\{1,a\},\{2,a\},\{2,b\},\{3,b\},\{3,c\},\{4,c\},\{4,d\},\{5,d\},\{5,e\},\{6,e\}\,\}.$$
    It is easy to see that $G=G(L,w)$ with $w=a1b2ac3bd4ce5d6e$. Hence, $G$ is convex. Let $H=(A \cup B, E')$ be the bipartite complement of $G$. We want to show that $H$ is no convex graph. Clearly, $A$ and $B$ are the partition classes of $H$ and there are no other ways to partition these vertices. By \autoref{lem:uniformword}, there exists a $w' \in (A\cup B)^*$ with $H=G(L,w')$ and $\vert w'\vert_{v}\in \{1,2\} $ for each $v \in A \cup B$. Since $A,B$ is a unique partition for $H$, for all $v\in A$ and $u \in B$, $\vert w'\vert_v \neq \vert v \vert_u$.
    
    We first assume  $\vert w'\vert_v=1$ for each $v \in A$. We start by considering $h_{a,1}(w')=001= h_{a,1}(w')$. By \autoref{cor:L-symmetry}, $h_{a,1}(w')=100= h_{a,1}(w')$ works analogously. Since $N_H(a)=\{3,4,5,6\}$, for each $v \in \{ 3,4,5,6\}$ and $j \in \{1,2\}$, $h_{v,j}(w')=01$. $3\notin \{1,4,5,6\} =N_H(b)$ implies $h_{\{a,b,1,2,3,5\}}(w')=a3b5a1b2$. By $3,5,2\in N_H(e)$, $h_{\{b,e,1,2,3,5\}}(w')=e3b5a1be$ holds, which contradicts $\{5,e\}\in E'$.

    So consider $h_{\{a,1,2\}}(w')=1aa2$ (note that $h_{\{a,1,2\}}(w')=2aa1$ works analogously). Since $\{a,5\}\in E'$ and $1,2\in N_H(e)$, $h_{\{a,e,1,2,5\}}(w')=e2a5a1e$, which contradicts $\{e,5\}\notin E'$. 

    Now we assume that for each $v \in A$, $\vert w'\vert_v=2$. Consider $h_{\{a,b,2\}}(w')=a22b$  ($h_{\{a,b,2\}}(w')=b22a$ is analogously shown). By $a,b\in N_H(6)$, $h_{\{a,b,2,6\}}(w')=6a22b6$. This contradicts $\{e,6\}\notin E'$, as $\{2,e\}\in E'$ and so $h_{\{a,b,e,2,6\}}(w')=6a2e2b6$. 

    Therefore, assume $h_{2,a}(w') = 001 = h_{2,b}(w')$ ($h_{2,a}(w')=100= h_{2,b}(w')$ is analogously proven). $\{b,3\},\{c,3\} \notin E'$ and $\{a,3\},\{e,3\},\{e,2\}\in E'$ together imply $h_{\{a,b,c,e,2,3\}}(w') = 2c3e2a3b$. Since $a,b\in N_H(6)$, $h_{\{b,c,e,6\}}(w') = 6ceb6$, which contradicts $\{6,e\}\notin E'$. 

    As there is no case remaining, $G$ is convex but its bi-complement is not.
%\qed
\end{proof}
\longversion{\subsection{2-Uniform Languages}}\shortversion{\paragraph*{2-Uniform Languages.}}
We will now study all 2-uniform languages. As before this covers all (seven) cases, which can be verified by looking at \autoref{tab:graphClasses-intro}.
\begin{toappendix}
The following Lemma shows (in view of \autoref{exa:represented-graphs-CompleteUnionNullGraphs}) that there are many ways to describe the same graph class, and the reader should find it easy to find many more representations of this class.

\begin{lemma}\label{lem:00shuffle11}
For $L = \{0,1\}^{2\text{-uni}}$, %\{ w \in \{ 0, 1 \}^* \; | \; |w|_0 = |w|_1 = 2 \}$, 
$\mathcal{G}_{L}=\{K_n\cup N_m\mid n,m\in\N\}$.%\todohfInPlace{Similar to \autoref{exa:represented-graphs}?}
%is exactly the set of graphs $G = (V,E)$ such that there exists a subset $V' \subseteq V$ with $E = \{ \{u,v\} \; | \; u,v \in V' \} $. Hence, 
\end{lemma}
\begin{proof}
\begin{enumerate}
    \item Let $G = (V,E) \in \cG_L$. Hence, there exists a word $w \in V^*$ such that $G = G(L,w)$. Let $V' = \{ v \in V \mid |w|_v = 2 \}$ and $E' = \{ \{u,v\} \mid u,v \in V' \}$.
\begin{align*}
\{ u, v \} \in E  \iff h_{u,v}(w) \in L 
                  \iff |w|_u = |w|_v = 2  \iff u, v \in V' \longversion{\\} \iff \{ u, v \} \in E'\,.\longversion{\phantom{XXXXXXXXXXXXXXXXXiii}}
\end{align*}
This shows $E = E'$.
\item Let $G = (V,E)$ be such that there exists a subset $V' \subseteq V$ with $E = \{ \{u,v\} \mid u,v \in V' \} $. Let $V' = \{ v_1, \dots, v_m \}$ and $V \setminus V' = \{ v_{m+1}, \dots, v_n \}$. 
Choose $w = v_1 \dots v_m \cdot v_1 \dots v_m \cdot v_{m+1} \dots v_n $. Let $G_2 = G(L,w)$. We prove that $G_2 = G$. Obviously, $V(G_2) = V$. 
\begin{align*}
\{ u, v \} \in E(G_2)   \iff h_{u,v}(w) \in L \iff |w|_u = |w|_v = 2 \iff u, v \in V' 
                   \longversion{\\}  \iff \{ u, v \} \in E\,.\longversion{\phantom{XXXXXXXXXXXXXXXXXXi}}
\end{align*}
This shows that $G(L,w) = G_2 = G$.
\end{enumerate}
Hence, $\mathcal{G}_{L}=\{K_n\cup N_m\mid n,m\in\N\}$.%\qed
\end{proof}
    
\end{toappendix}

We next provide a collection of useful closure properties that allow us to apply \longversion{the Trash }\autoref{trash-theorem} \longversion{in a number of circumstances}\shortversion{frequently}. Recall the definition of $T_L$ from that theorem.

\begin{prprep}\applabel{prop:2-uniform-plus-vertices}
Let $L\subseteq\{0,1\}^*$ be 2-uniform. Assume that for each $G=(V,E)\in \cG_L$, there exists a 2-uniform word~$w\in V^*$ such that $G=G(L,w)$.\\
1. If $L\cap \{0110,0011\}\neq\emptyset$, then 
$\cG_L$ is closed under adding universal vertices. \\ 
2. If ${L}\cap \{0110,0011\}\neq\emptyset$, then  $\cG_{\overline{L}}$  is closed under adding isolated vertices.
\end{prprep}
\begin{proof}
Let $L\subseteq\{0,1\}^*$ be 2-uniform.\\
1. Let $G=(V,E)\in \mathcal{G}_{L}$. Hence, there is a $2$-uniform $w\in V^*$ with $G=G(L,w)$. Let $\ta\notin V$. First, assume $0110\in L$.  Consider $w'=\ta w\ta$. Then, $\ta$ is a universal vertex in $G'=G(L,w')$ and $G=G'[V]$, because $w$ is $2$-uniform. Otherwise, by our assumption $0011\in L$.   Consider $w''=w\cdot \ta^2$. Then, $\ta$ is a universal vertex in $G''=G(L,w'')$ and $G=G''[V]$.\\
2. By the previous item,  $\cG_L$ is closed under adding universal vertices. Hence, its corresponding class of complement graphs ${\cG_L'}$ is closed under adding isolated vertices. By \autoref{prop:compl}, $\cG_{\overline{L}}={\cG_L'}$. 
\qed 
\renewcommand{\qed}{}
\end{proof}

\begin{toappendix}
\noindent
We can add some more properties to the list of \autoref{prop:2-uniform-plus-vertices}.
\begin{remark}
Now assume that for each $G=(V,E)\in \cG_{\overline{L}\cap \{0,1\}^{2\text{-uni}}}$, there exists a 2-uniform word~$w\in V^*$ such that $G=G(\overline{L}\cap \{0,1\}^{2\text{-uni}},w)$.
\begin{enumerate} \addtocounter{enumi}{2}
    \item If $\overline{L}\cap \{0110,0011\}\neq\emptyset$, then  $\cG_{\overline{L}\cap \{0,1\}^{2\text{-uni}}}$ is closed  under adding universal vertices.
    \item If $\overline{L}\cap \{0110,0011\}\neq\emptyset$, then 
    $\cG_{{L}\cup T_L}$ is closed under adding isolated vertices.
\end{enumerate} 
\end{remark}
\begin{proof}
\begin{enumerate} \addtocounter{enumi}{2}
   \item As $\widehat L\coloneqq {\overline{L}\cap \{0,1\}^{2\text{-uni}}}$ is 2-uniform, the claim follows with the first item, where $\widehat L$ takes on the role of~$L$, as $\overline{L}\cap \{0110,0011\}=\widehat L\cap \{0110,0011\}$.
    \item As $T_L=\overline{\{0,1\}^{2\text{-uni}}}$, this follows by the third item by the reasoning of the second item (using \autoref{prop:compl}) with De Morgan's Law as $\overline{{\overline{L}\cap \{0,1\}^{2\text{-uni}}}} = L \cup T_L$. %\todoscsInPlace{I don't understand, what $T_{\hat L}$ has to do with this. I would write that $\overline{{\overline{L}\cap (0^2\shuffle 1^2)}} = L \cup T_L$ and reference \autoref{prop:compl} again.}
%    \item  Let $G=(V,E)\in \mathcal{G}_{L\cup T_L}$.  Hence, there is a $w\in V^*$ with $G=G(L\cup T_L,w)$. Let $\ta\notin V$. First, assume $0110\notin L$. Hence, $0110\notin L\cup T_L$.  Consider $w'=\ta w\ta$. Then, $\ta$ is an isolated vertex in $G'=G(L\cup T_L,w')$ and $G=G'[V]$.  Otherwise, by our assumption $0011\notin L$.   Consider $w''=w\cdot \ta^2$. Then, $\ta$ is an isolated vertex in $G''=G(L\cup T_L,w'')$ and $G=G''[V]$.
\qed 
\end{enumerate}
\renewcommand{\qed}{}
\end{proof}
\end{toappendix}

The discussions in the paragraph preceding \autoref{lem:uniformword} are also related in particular to the last statement of the previous proposition.
After this type of prelude, well-known non-trivial graph classes are characterized in the following. We start with the classes of (co-)interval graphs.

\begin{theorem}\label{thm:intervalgraphs}
$\mathcal{G}_{\langle 0101,0110\rangle}$ and 
$\mathcal{G}_{\overline{ \langle 0011 \rangle }}$ is the class of interval graphs. 
\end{theorem}

\begin{proof}
Let $L = \langle 0101, 0110 \rangle$. Using the notation from \autoref{trash-theorem},\longversion{ we get} $\hat L = \overline{ \langle 0011 \rangle }$. \longversion{According to}\shortversion{By} \autoref{thm:interval-model}, $\mathcal{G}_{L}$ is the class of interval graphs and hence $\mathcal{G}_{L}$ is closed under adding universal vertices. 
Due to \autoref{lem:uniformword}, we can apply\longversion{ the second part of} \autoref{prop:2-uniform-plus-vertices}\shortversion{.2} to\longversion{ the language} $\langle 0011 \rangle$. 
%We will prove that $\mathcal{G}_{\hat L}$ is closed under adding isolated vertices. For this consider a graph $G \in \mathcal{G}_{\hat L}$. There is a word $w \in V(G)^\ast$ such that $G(\hat L,w) = G$. Let $u$ be a vertex such that $u \notin V(G)$. The graph $G(\hat L,w \cdot uu) $ is the graph $G$ with $u$ added as an isolated vertex. 
Thus, $\mathcal{G}_{\hat L}$ is closed under adding isolate\longversion{d vertice}s. Hence, \longversion{we get the equality }$G_L = G_{\hat L}$ from \autoref{trash-theorem}. %\qed 
\end{proof}
%\begin{toappendix}
\begin{remark}
The characterization result of interval graphs by  $L = \langle 0101,0110\rangle$ also follows from the characterization of interval graphs as 1-11-representable graphs that can be represented by uniform words containing two copies of each letter, see~\cite{CheKKKP2019}. However, this is also \longversion{a direct consequence of}\shortversion{follows directly from} our interval representation model with respect to $\langle 0101,0110\rangle$ that also supports the intuition that the pattern $0101$ covers the case of proper interval overlap, while $0110$ covers the case of containment. From a graph-theoretic point of view, looking at their intersection graph models, a natural generalization of interval graphs are chordal graphs. \autoref{thm:intervalgraphs} shows (in particular) that interval graphs can be represented by a finite language. However, \cite{FenFFKS2026} shows that chordal graphs cannot be represented by any finite language.\longversion{ The related information-theoretic argument also proves that there must be chordal graphs which are not interval graphs. This fact is of course well-known in the graph-theory community, but this argument is clearly different from the standard proofs, see \cite{Gol2004i,Gol2004l}.}
\end{remark}

\noindent
By \autoref{prop:compl}, we obtain:
\begin{corollary} \label{cor:co-intervalgraphs}
For $L = \langle 0011 \rangle $, $\mathcal{G}_{L}$ is the class of co-interval graphs. 
\end{corollary}

\begin{remark}
According to \autoref{thm:interval-model} the interval model with respect to the language $\langle 0011\rangle$ is a geometric model for the class of co-interval graphs.
\end{remark}

\begin{comment}
\begin{lemma} %\label{lem:0011appear2} removed because \autoref{lem:uniformword} is more general
    Let $L$ be a 2-uniform \todoscsInPlace{Does 2-uniform imply that $L \neq \emptyset$?} language with $\langle 1001\rangle \cap L = \emptyset$ \todoscsInPlace{So $L \subseteq \langle 0011, 0101 \rangle$? Wouldn't that description be easier?} and $G\in \mathcal{G}_{L}$. Then there exists a $w\in V^{\ast}$ such that $G=G(L,w)$ and $\vert w\vert_v=2$ for each $v\in V$.
\end{lemma}
\begin{proof}
    Let $G=(V,E)\in \mathcal{G}_{L}$. Then there exists $w\in V^{\ast}$ with $G=G(L,w)$. Assume there exists a $v\in V$ with $\vert w\vert_v\neq 2$. Hence, $v$ is independent in $G$. We want to construct a $w'\in V^{\ast}$ with $G=G(L, w')$, $\vert w'\vert_v = 2$ and $\vert w'\vert_u = \vert w\vert_u$ for each $u\in V\setminus\{v\}$. Let $w''$ be $w$ after deleting each appearance of $v$. Define $w'=vw''v$ and $G'=(V,E')=G(L, w')$. 
    
    Clearly, for $s,t \in V\setminus \{v\}$, $\{s,t\}\in E$ \iffl $\{s,t\}\in E'$, as $h_{s,t}(w)=h_{s,t}(w')$.
    Recall, $v$ in independent in $G$. Since $v$ is only at the beginning and end of $w'$, $h_{v,u}(w')\in L(0(1)\ast0)$ \todoscsInPlace{How do we want to write regular expressions? I believe $0(1)^\ast 0$ would be more consistent with the literature.} for all $s\in V\setminus \{v\}$. Hence $v$ is independent in $G'$ and $G=G'$.    
\end{proof}
\end{comment}

\longversion{The next rather famous class of graphs that we study are permutation graphs.
According to graphclasses.org,}\shortversion{\noindent As} permutation graphs are\longversion{ exactly} the interval containment graphs, \autoref{thm:interval-model} \longversion{gives an interpretation of our intersection model and hence }yields: 
%As we could not trace back a proof for this statement in the literature, we are providing a short one next, using our general geometric interpretations of 2-uniform languages and the corresponding graph classes.

\begin{theorem} \label{thm:permGraphs}
For $L = \langle 0110 \rangle$, $\mathcal{G}_L$ is the class of permutation graphs.
\end{theorem}

\begin{toappendix}
The following lemma should be known for permutation graphs, but we have made this explicit for two reasons: first, such lemmas (or observations) are hard to find in the literature and secondly, it is an example how to reason within our framework in a very concise way, using \autoref{prop:2-uniform-plus-vertices} and \autoref{lem:uniformword}.

\begin{lemma}\label{lem:0110-plus-universal-vertices}
$\cG_{\langle 0110 \rangle}$ is closed under adding universal vertices.
\end{lemma}
    
\begin{proof}
Let $G\in \mathcal{G}_{\langle 0110 \rangle}$. Hence, there is a $w\in V(G)^*$ with $G=G(\langle 0110 \rangle,w)$. Consider $w'=\ta w\ta$ for some $\ta\notin V(G)$. Then, $\ta$ is a universal vertex in $G'=G(\langle 0110 \rangle,w')$ and $G=G'[V(G)]$. %\qed
\end{proof}
\end{toappendix}

%\todohfInPlace{I have put some small things under comment as I deem them superfluous or wrong.}
%\begin{remark}
%According to \autoref{thm:interval-model} the interval model with respect to the language $\langle 0110\rangle$ is an models for the class of permutation graphs.
%\end{remark} \todo{Ref?}

\begin{comment}
\begin{theorem}$\langle 0110\rangle$-representable graphs are 
$C_5$-free.\todohfInPlace{Maybe, not needed by previous arguments, or a nice example of a `simplified argument'.}
\end{theorem}

\begin{proof}
By \autoref{thm:hereditary}, it suffices to show that $C_5$ is not $\langle 0110\rangle$-representable.\todo{details ...}
\end{proof}
\end{comment}

\begin{lemrep} \applabel{lem:permGraphs2} For $\hat L = \langle 0110 \rangle  \cup  \overline{\{0,1\}^{2\text{-uni}}}$,
$\cG_{\hat L}$ are the permutation graphs.
\end{lemrep}

\begin{proof} 
Let $L=\langle 0110 \rangle$.
Observe that $T_L=\overline{\{0,1\}^{2\text{-uni}}}$ in the sense of \autoref{trash-theorem} as $L$ is 2-uniform, so that  $\hat L =L\cup T_L$.
By \autoref{trash-theorem}, combined with the closure of permutation graphs under adding universal vertices (see \autoref{lem:0110-plus-universal-vertices}),   $ \cG_{\hat L}= \cG_L$ follows.%\qed
\end{proof}

The class of permutation graphs are closed under taking graph complements.  Therefore, \autoref{prop:compl} and \autoref{lem:permGraphs2} imply the following:
\begin{corollary}\label{cor:co-permutation}
$\mathcal{G}_{\langle 0011, 0101 \rangle}$ is the class of permutation graphs.
\end{corollary}

%\todohfInPlace{Again, I have put some stuff under comment.}
\begin{comment}
\begin{theorem}$\langle 0110\rangle$-representable graphs are 
AT-free. \todoscs{How does that fit into what was already known about permutation graphs? I think this theorem should have some explanation why it is here.} \todohf{It might be an example of a `simplified argument'.}
\end{theorem}
As every cycle of length at least six contains an asteriodal triple, the last two theorems imply:

\begin{corollary} \label{cor:permGraphs3}
Any induced cycle in a $\langle 0110\rangle$-representable graph has length at most four.
In other words, $\langle 0110\rangle$-representable graphs are hole-free.
\end{corollary}
    
\end{comment}

%\todoscs{This was previously known under the name Alternance graphs, see for example Reducing prime graphs and recognizing circle graphs by Bouchet. Try to find the first paper with this idea! HF: Do you still want to find some French paper? Recall that at the very end, we have a subsection on circle graphs; there, you can also put some history.}

\begin{toappendix}
We now turn our attention to yet another well-studied graph class, namely circle graphs.
Kitaev and Lozin showed in~\cite[Theorem 5.1.7]{KitLoz2015} the next result that, reformulated in our framework, reads as follows:
\begin{proposition} \label{prop:circlegraphs}
$\mathcal{G}_{\langle 0101\rangle}$ is the class of circle graphs. 
\end{proposition}

We again also get this proposition as a corollary of \autoref{thm:interval-model}, because of the fact that circle graphs are exactly the (interval) overlap graphs, see \cite{Gol2004f}.
    
\end{toappendix}

\begin{thmrep} \applabel{thm:co-circle}
$\cG_{\langle 0011,0110 \rangle} = \{ G \cup N_m \mid G \text{ is a co-circle graph}, m \in \mathbb{N} \} $.
\end{thmrep}

\begin{proof}
Let  $L = \langle 0011,0110 \rangle$. We will show that $$\cG_{\overline{L}} = \{ G \nabla K_m \mid G \text{ is a circle graph}, m \in \mathbb{N} \}\,. $$ This is sufficient because of \autoref{prop:compl}. 
\begin{enumerate}
\item Let $G \in \cG_{\overline{L}}$. Hence, a word $w \in V(G)^\ast$ exists such that $G = G(\overline{L},w)$. Because of \autoref{lem:edgesets}, the following holds:
$E(G) = E(G(\langle 0101 \rangle),w) \cup \binom{V(G)}{2} \setminus E(G(\{0,1\}^{2\text{-uni}},w))$.
The graph $G(\langle 0101 \rangle),w)$ is a circle graph according to \autoref{prop:circlegraphs}. In the graph $G(\{0,1\}^{2\text{-uni}},w)$, the set $V_2$ is a clique and $V \setminus V_2$ is a set of isolated vertices. Hence, $G \in \{ G \nabla K_m \mid G \text{ is a circle graph}, m \in \mathbb{N} \}$.
\item Let $G' = G \nabla K$ such that $K$ is isomorphic to $K_m$ with $m \in \mathbb{N}$ and $G = (V,E)$ is a circle graph with $n$ vertices. According to \autoref{prop:circlegraphs}, a word $w$ exists such that $G = G(\langle 0101 \rangle,w)$. Let $V(K) = \{ u_1, \dots u_m \}$ and $w' = w \cdot u_1 \cdots u_m$. According to \autoref{thm:graph-decomposition}, $$E(G(\overline{L}, w')) = E_{1,1}(\overline{L}, w') \cup E_{1,2}(\overline{L}, w') \cup E_{2,2}(\overline{L}, w')\,.$$ Clearly, $E_{1,1}(\overline{L}, w') = \binom{V(K)}{2}$, $E_{1,2}(\overline{L}, w') = \{ \{ v, u \} \mid v \in V, u \in V(K) \}$ and $(V_{2,2}(\overline{L}, w'), E_{2,2}(\overline{L}, w')) = G $. Hence, $G(\overline{L}, w') = G'$ and $G' \in \cG_{\overline{L}}$.\qed
\end{enumerate}\renewcommand{\qed}{}
\end{proof}

Note that we cannot prove \autoref{thm:co-circle} using \autoref{prop:2-uniform-plus-vertices}, because not every graph in $\cG_{\langle 0011,0110 \rangle}$ can be represented by a 2-uniform word (see \autoref{exm:verylongexample}). 

\begin{comment}
So, bipartite $\langle 0110\rangle$-representable graphs are chordal bipartite.
This raises two questions:\todohfInPlace{I suppose this is a relic from times where we did not know the permutation graph characterization?}
\begin{enumerate}
    \item Are there any graphs that are AT-free and $C_5$-free (hole-free) and that are not $\langle 0110\rangle$-representable?
    \item Are there chordal bipartite graphs that are not $\langle 0110\rangle$-representable?
\end{enumerate}
\end{comment}

\begin{toappendix}

\subsection{Languages With At Most Two Occurrences of $0$ or $1$}
\label{subsec:atmost-two-occurrences} 

As we have seen in the preceding subsections, already the (nearly) length-uniform languages themselves offer quite a rich spectrum of graph classes, many of them well-known from other contexts. In this subsection, we only study some of the graph classes that can be obtained from languages with at most two occurrences of $0$ or~$1$. We will encounter two different graph classes this way.
Our first example was already studied in \autoref{exa:threshold}. %Now, we give an exact classification of the class.

In \cite{KlaPet87}, Klavžar and Petkovšek discussed intersection graphs of halflines, i.e., \longversion{one-side bounded}\shortversion{unbounded} intervals. These are a special class of cobipartite graphs. The following theorem might look a bit awkward in its formulation, but it gives a nice characterization ``up to isolates'', and similar exceptions can also be found in \autoref{tab:graphClasses-intro}. Again, this is related to the fact that the class of halfline intersection graphs is not closed under adding isolates, see \autoref{propos:closure-adding-isolates}.

\begin{theorem} \label{thm:intersection-graph-halflines} Let $L_{\text{HL}}\coloneqq\langle 01,011,0101,0011,0110\rangle$.
    A graph $G=(V,E)$ is $L_{\text{HL}}$-representable \iffl $G$ is the disjoint union of $G_1$ and $G_2$ where $G_1$ is an intersection graph of halflines and $G_2$ is a null graph.
\end{theorem}
%\todoscs{Reviewer 1: Theorem 5.28 ignores the fact that intersection graphs of halflines are to cobipartite graphs as chain graphs are to bipartite graphs (that is, they are the complements of chain graphs -- see e.g. the references mentioned in https://www.sciencedirect.com/science/article/pii/S0012365X01000942).}

\begin{proof}
%
%\begin{example}
%    Define $L\coloneqq  \langle011\rangle \cup \bigcup_{(i,j)\in {\mathbb{N}^+}^2\setminus \{(1,2),(2,1)\}}1^i \shuffle 1^j$. We now will show that $\mathcal{G}_L$ is the set of intersection graph of halflines.
%
Let $G=(V,E)$ be an intersection graph of half-lines of order $n$ with the interval representation $\mathcal{I}=\{I_v\mid v\in V\}$. Furthermore, $V_1\subseteq V$ denotes the set such that for each $v\in V_1$ there exists an $a_v \in \mathbb{R}$ with $I_v = [a_v , \infty)$, i.e., $I_v$ is a left-bounded half-line. $V_2\subseteq V$ is the set of vertices $v\in V$  such that there is an $a_v \in \mathbb{R}$ with $I_v = ( -\infty, a_v ]$, i.e., $I_v$ is a right-bounded half-line. Thus $V_1,V_2$ are cliques in~$G$. Define the an order  on~$V$ by $$\mbox{$\preceq'$}\coloneqq \{ (v,u)\in V^2 \mid (a_v<a_u)\vee (a_v=a_u \wedge v\in V_1 \wedge u\in V_2) \}\,.$$ Extend $\preceq'$ to a linear ordering that we write as~$\preceq$. Let $w'$ be the word that enumerates~$V$ in the ordering~$\preceq$, i.e., $|w'|=n$. Let $w$ denote the word $w'w''$ where $w''$ is an enumeration of~$V_2$. Define $G'=(V,E')=G(L_{\text{HL}},w)$. Clearly, for each $i\in \{1,2\}$ and $v\in V_i$, $v$ appears~$i$ times in~$w$. As $L_{\text{HL}}$ contains all 1-uniform and 2-uniform words,  $V_1$ and $V_2$ are cliques in~$G'$. This leaves to show that for all  $v\in V_1, u\in V_2$, $\{v,u\}\in E$ \iffl $h_{v,u}(w)=011$. By definition of~$w$, the last position of $h_{v,u}(w)$ is~1. Thus, $\{v,u\}\in E$ \iffl $a_v \leq a_u$ \iffl $h_{v,u}(w)=011$. Thus, $G=G'$. As in \autoref{propos:closure-adding-isolates}, we can easily add isolated vertices to~$G'$.

Now let $G=(V,E)\in \mathcal{G}_{L_{\text{HL}}}$. Then there exists a word $w=w_1\cdots w_{\ell}\in V^*$ with $w_i\in V$ for $i\in [\ell]$ such that $G=G(L_{\text{HL}},w)$. Let $V_j\coloneqq \{v\in V \mid \vert w\vert_v = j\}$ for $j \in \{1,2\}$. Define 
    $$a_v\coloneqq \begin{cases}
        \min \{i\mid w_i=v\}, & v\in V_1 \cup V_2,\\
        0, & \text{otherwise,}
    \end{cases} \text{ and }I_v\coloneqq \begin{cases}
        (-\infty, a_v], & v\in V_2,\\
        [a_v, \infty), & v\in V_1\,.
    \end{cases}$$
    Furthermore, $G'=(V,E')$ denotes the intersection graph of $\mathcal{I}=\{I_v\mid v\in V\}$.  Now we want to show $E=E'$. Let $v,u\in V$. If $v \notin V_1 \cup V_2$, then $v$ is isolated in $G$. %If $u \in V_2$\todohfInPlace{and $v\in V_2$?}, then $0\in I_v\cap I_u$. Otherwise\todohfInPlace{If $u,v\in V_1$?}, $\ell \in I_v\cap I_u$. 
    Thus, we can assume $v,u \in V_1 \cup V_2$. If $v,u\in V_1$ (respectively $v,u\in V_2$) then $\{v,u\}\in E$ and $\ell \in I_v\cap I_u$ (respectively $0\in I_v \cap I_u$). 
    So, we can assume $v\in V_1$ and $u\in V_2$. $\{v,u\}\in E$ if an only if $a_v\leq a_u$ \iffl $ I_v \cap I_u \neq \emptyset$.
%\end{example}
\end{proof} 

We can characterize the class of half-line intersection graphs by the infinite language $L\coloneqq \langle011\rangle \cup \bigcup_{(i,j)\in \mathbb{N}_{\geq 1}^2\setminus \{(1,2),(2,1)\}}\{0,1\}^{i,j\text{-uni}}$. The proof is nearly identical to the previous one, apart from the discussion of isolated vertices. If a vertex (letter) has frequentness different from $1$ or~$2$, then this corresponds to the interval $[0,\infty)$.

Our last theorem of this subsection does not really fit here, as we also see vertices of frequentness~3. However, it shows that one can still expect new graph classes to show up when increasing frequentnesses further. At least, we saw no way to describe this graph class with vertices of frequentness at most~2. Moreover, this theorem is needed later to disprove a conjecture that one might tend to believe.

We also need to introduce a further graph class first. 
A bipartite graph $G=(V,E)$ with the partition classes $A$ and $B$ is called an \emph{interval bigraph} if there is family $\mathcal{I} = \{ I_v =[l_v,r_v] \}_{v\in V}$ of closed intervals such that, for every $u \in A$ and $v\in B$, $I_u \cap I_v \neq \emptyset$ \iffl $\{ u,v \} \in E$.  Keep in mind that $I_u \cap I_v \neq \emptyset$ holds \iffl $\max\{l_v,l_u\}\leq \min\{r_v,r_u\}$.
This graph class was introduced in~\cite{HarKabMcM82}.
There are very close connections to interval digraphs~\cite{SenDRW89}, as explained in~\cite{Mul97}.
We also refer to \cite{DasSah2024}.  

\begin{theorem} \label{thm:intervalBigraphs}
    $\mathcal{G}_{\langle 01110,01101,01011,01100,01010,01001\rangle}$ is the set of interval bigraphs.
\end{theorem}

\begin{proof}
    To simplify the notation, define $$L \coloneqq \langle 01110,01101, 01011,01100,01010,01001\rangle\,.$$ Note that the words $w$ with $\{ |w|_0, |w|_1 \} = \{ 2, 3 \}$ missing in $L$ are exactly the words that start with two identical letters. Let $G = (V,E)\in \mathcal{G}_L$. Then there exists a $w\in V^*$ with $G=G(L,w)$. Assume there exists a $v\in V$ with $\vert w\vert_v \notin \{2,3\}$. Then define $w'\coloneqq vv\cdot h_{V\setminus\{v\}}(w)$. % where $w''$ is $w$ after deleting each appearance of $v$. 
    Note that $v$ is an isolated vertex in~$G$  and $G(L,w')$. Since $h_{x,y}(w)=h_{x,y}(w')$ for all $x,y\in V\setminus \{v\}$, $G=G(L,w')$. Thus, we can assume that each $v\in V$ appears exactly twice or thrice in~$w$.

    %Define $A\coloneqq \{a\in V \mid \vert w\vert_a=2\}$ and $B\coloneqq \{b\in V \mid \vert w\vert_b=3\}$.\todoscsInPlace{I changed the notation from $A$ and $B$ to $V_2$ and $V_3$ as this is already defined.} 
    Clearly, $G$ is bipartite with the classes $V_2$ and $V_3$. Let $w=w_1\cdots w_{\ell}$ with $w_1,\dots, w_{\ell}\in V$. We denote by $i_{v,j}\in[\ell]$ for $v\in V$ and $j\in \{1,2\}$ the number such that $w_{i_{v,j}}=v$ and $\vert w_1 \cdots w_{i_{v,j}}\vert_v = j$, i.e., $i_{v,j}$ is the index of the $j^{\text{th}}$ occurrence of~$v$ in~$w$.

    For $v\in V$, define $I_v=[i_{v,1},i_{v,2}]$. Let $\ta\in V_2$ and $\tb \in V_3$. Consider $\{\ta,\tb\}\notin E$. This holds \iffl $h_{\ta,\tb}(w) \in \{00111,11001,11010,11100\}$. This are exactly the cases where $I_\ta \cap I_\tb = \emptyset$. Therefore, $G$ is an interval bigraph. %\todoscsInPlace{This doesn't hold if we use the proof for the second corollary below. (See comment there!)}

    Now let $G=(V,E)$ be an interval bigraph of order~$n$ with the classes $A,B \subseteq V$ and the family of intervals $\mathcal{I}=\{I_v\coloneqq [l_v,r_v]\}_{v \in V}$. %In the proof of Lemma 8 of \cite{DasSah2024}\todoscsInPlace{Reviewer 1: There is no need for the reference to the proof of Lemma 8 of [16]; the fact we can assume the endpoints are different is a standard perturbation argument, and if there is to be a citation, it should be to a textbook or one of the old papers on interval graphs.} it is proven that
    By a standard perturbation argument, we can assume the endpoints of the intervals are different. Define $x_1,\dots,x_{2n}\in \{l_v,r_v\mid v\in V\}$ such that for $i,j\in [2n]$, $x_i<x_j$ \iffl $i<j$. Furthermore define $w'\coloneqq w_1\cdots w_{2n} \in V^{2n}$ with $w_i=v$ if $x_i\in \{l_v,r_v\}$ for $v\in V$. By adding each vertex from~$A$ once in an arbitrary order at the end of~$w'$, we obtain $w$. Clearly, each vertex from~$A$ appears three times and each from~$B$ two times in~$w$. Define $G'=(V,E')=G(V,w)$. Let $v\in A$ and $u\in B$. Consider $\{v,u\}\in E$. Hence $\max\{l_v,l_u\}\leq \min\{r_v,r_u\}$. This implies $h_{v,u}(w)  \in \{01010, 01100, 10010, 10100\}\subseteq L$ and $\{v,u\}\in E'$.

    Assume $\{v,u\}\in E'$. Since at the end of~$w$, $A$ is enumerated, there the last symbol of $h_{v,u}(w)$ is a~$0$. Hence, $h_{v,u}(w) \in \{01010,01100,10010,10100\}$. In this case  $\max\{l_v,l_u\}\leq \min\{r_v,r_u\}$ and $\{u,v\}\in E$. Thus, $G=G'$.%\qed
\end{proof} 

If we use the first and the last occurrence of a letter to describe the bounds of an interval (instead of the first two occurrences), we get the following corollary: 

\begin{corollary}
   $\mathcal{G}_{\langle \{0,1\}^{2,3\text{-uni}} \setminus \{ 00011, 11000\} \rangle}$ is the set of interval bigraphs. 
\end{corollary}
    
\end{toappendix}
\begin{comment}
\noindent
By the proof of the previous theorem, we also get:\todohfInPlace{check this}
\todoscsInPlace{I believe this is not true for the third paragraph.}
\begin{corollary}
   $\mathcal{G}_{\langle 01101,01011,01100,01010\rangle}$ is the set of interval bigraphs.
\end{corollary}

This is remarkable, as $\langle 01101,01011,01100,01010\rangle=\langle 01010,01100,10010,10100\rangle$, i.e., we have the word interpretation of the interval model, concatenated by a single~$0$. Of course, we can consider similar modifications to all 2-uniform languages that we looked at so far and we would obtain some bipartite analogues.
For instance, we could take $\langle 01010,10100\rangle=\langle 01010,01011\rangle$ as a bipartite overlap graph model. We are not aware of any studies of such a model (possibly under a different name).
    
\end{comment}

\section{Limitations of Language-Representability}

Next, we derive another nice closure property that all language-representable graph classes enjoy.
%; see \autoref{fig:twin} for an illustration.  

\begin{definition}\label{deftwinning}
We say that a graph class $\cG$ is %\emph{closed under adding some twins} or 
\emph{closed under twinning}\longversion{\footnote{This operation might look weird at first glance, but a very similar one appeared in the study of word-representability of split graphs, see, e.g., \cite{CheKitSai2022}.}} if for every $G=(V,E)\in \cG$ and for every $v\in V$ and some $v'\notin V$, at least one of the following two statements is true:
\begin{itemize}
    \item The graph $G'_v$ obtained from $G$ by adding $v'$ as a true twin of~$v$ belongs to $\cG$.
    \item The graph $G''_v$ obtained from $G$ by adding $v'$ as a false twin of~$v$ belongs to~$\cG$.
\end{itemize}
\end{definition}

\noindent
Notice that $\{G'_v,G''_v\}\subseteq\cG$ is possible in \autoref{deftwinning}.

\begin{theorem}\label{thm:twins}
Let $L\subseteq\{0,1\}^*$. Then, $\cG_L$ is closed under twinning.
\end{theorem}

\begin{proof}
Fix $L\subseteq\{0,1\}^*$. Let $G=(V,E)\in \cG_L$. Hence, there is some word~$w$ over the alphabet~$V$ such that $G=G(L,w)$. Let $v\in V$.
Then, we can decompose $w=w_0vw_1v\cdots vw_k$ such that, for each $i\in\N$, $i\leq k$, $v\notin\alph{(w_i)}$. Let $v'\notin V$. Define
%\begin{equation}
$w'\coloneqq w_0vv'w_1vv'\cdots vv'w_k\,.$
%\label{eq:twin-def}
%\end{equation}
Let $G'=(V',E')=G(L,w')$ such that $V'=V\cup\{v'\}$. 
By construction, for any $x,y\in V$, $\{x,y\}\in E\iff \{x,y\}\in E'$.
Moreover, for any $x\in V\setminus\{v\}$, $\{x,v\}\in E\iff \{x,v'\}\in E'$. This means that $v,v'$ are twins. Whether or not they are true or false twins depends on whether or not $h_{v,v'}(w')\in L$.
In any case, $G'\in\cG_L$.%\qed 
\end{proof}

\begin{figure}[bt]

\begin{subfigure}[t]{.51\textwidth}
        \centering
        \begin{tikzpicture}
            \node (a) {$\ta$};
            \node[below right= 0.6cm and 0.1cm of a] (e) {$\te$};
            \node[right of=a] (f) {$\tf$};
            \node[above of=f] (b) {$\tb$};
            \node[right of=f] (c) {$\tc$};
            \node[above right= 0.1cm and 0.3cm of c,color=red] (c') {$\tc'$};
            
            \node[below left=0.6cm and 0.1cm of c] (d) {$\td$};

            \draw (a) -- (b);
            \draw (a) -- (e);
            \draw (a) -- (f);
            \draw (b) -- (c);
            \draw (b)[color=red] -- (c');
            \draw (c')[color=cyan] -- (c);
            \draw (b) -- (f);
            \draw (c) -- (d);
            \draw (c) -- (f);
            \draw (c')[color=red] -- (d);
            \draw (c')[color=red] -- (f);
            \draw (d) -- (e);
            \draw (d) -- (f);
            \draw (e) -- (f);
        \end{tikzpicture}
%        \caption{\mbox{$G(L,w)$}}
\end{subfigure}\qquad
\begin{subfigure}[t]{.43\textwidth}
            \centering
        \begin{tikzpicture}
            \node (a) {$\ta$};
            \node[below right= 0.6cm and 0.1cm of a] (e) {$\te$};
            \node[right of=a] (f) {$\tf$};
            \node[above of=f] (b) {$\tb$};
            \node[right of=f] (c) {$\tc$};
            \node[above right= 0.15cm and 0.3cm of c,color=red] (c') {$\tc'$};
            
            \node[below left=0.6cm and 0.1cm of c] (d) {$\td$};

            \draw (a) -- (b);
            \draw (a) -- (e);
            \draw (a) -- (f);
            \draw (b) -- (c);
            \draw (b)[color=red] -- (c');
            \draw (b) -- (f);
            \draw (c) -- (d);
            \draw (c) -- (f);
            \draw (c')[color=red] -- (d);
            \draw (c')[color=red] -- (f);
            \draw (d) -- (e);
            \draw (d) -- (f);
            \draw (e) -- (f);
        \end{tikzpicture}
    \end{subfigure}
    \caption{$L=\{ 0101, 1010, 001, 110 \}$, $w=\mathtt{aebac{\color{red}c'}bdc{\color{red}c'}edf}$\shortversion{,}\longversion{ and} $w'=\mathtt{aebac{\color{red}c'}bd{\color{red}c'}cedf}$.}\label{fig:twin}
\end{figure}
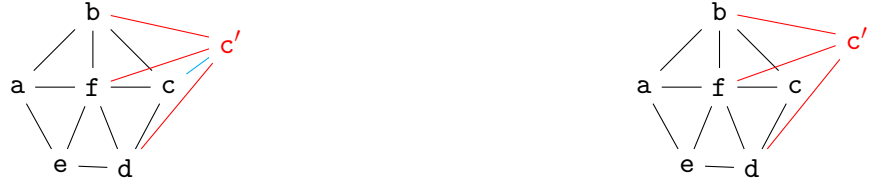

\begin{toappendix}
\begin{remark}
It is worth noting how both adding true and false twins can result in a graph from our graph class. This is due to the fact that there is a certain arbitrariness in the definition of~$w'$ in the proof of \autoref{thm:twins}. More generally speaking, we could pick any $w'$ from $$\{w_0\}\cdot (v\shuffle v')\cdot \{w_1\}\cdot (v\shuffle v')\cdots (v\shuffle v')\cdot \{w_k\}\,,$$
where $\shuffle$ denotes the shuffle operation.
For some choice of $w'$, $h_{v,v'}(w')\in L$ might hold, while for other choices of $w'$, $h_{v,v'}(w')\notin L$.
\end{remark}    
\end{toappendix}
\noindent
We will derive a number of interesting consequences from twinning.

\begin{theorem}\label{thm:includingbigbiclique}
If $L\not\subseteq 0^*\cup 1^*$, then $\cG_L$ contains arbitrarily large complete bipartite graphs as subgraphs of graphs in $\cG_L$.
\end{theorem}

\begin{proof}
As $L\not\subseteq 0^*\cup 1^*$, there is some $G\in\cG_L$ that contains an edge. As $\cG_L$ is hereditary, $K_{1,1}\in \cG_L$. Let $x,y$ be its two vertices. 
By alternatingly adding a twin to $x$, then to $y$, then to $x$ again, etc., we obtain, after $2n$ twinnings, a graph $G_n$ that contains $K_{n+1,n+1}$ as a subgraph, irrespectively of whether false or true twins have been added. According to \autoref{thm:twins}, $G_n\in \cG_L$.
%\qed 
\end{proof}

\begin{remark}
This theorem implies that no sparse graph class can be represented by any language. This excludes, for instance, graphs of bounded degeneracy (in particular planar graphs or forests), or of bounded graph genus, or of bounded treewidth, or of bounded treedepth, just to mention some examples.
\end{remark}

\begin{toappendix}
Finally, we re-consider the parameter treewidth explicitly; similar considerations are possible with other sparseness parameters.

\begin{theorem}\label{thm:bounded-treewidth}
%\shortversion{$(*)$}
Let $L\subseteq\{0,1\}^*$. The following four conditions are equivalent.
\shortversion{(1) There is some $t\geq 0$ such that all graphs in $\cG_L$ have treewidth bounded by~$t$. (2) All graphs in $\cG_L$ have treewidth~$0$. (3) $\cG_L=\{N_n\mid n\in\N_{\geq1}\}$. (4) $L\subseteq 0^*\cup 1^*$.}
\longversion{\begin{enumerate}
    \item There is some $t\geq 0$ such that all graphs in $\cG_L$ have treewidth bounded by~$t$. \item  All graphs in $\cG_L$ have treewidth~$0$. \item $\cG_L=\{N_n\mid n\in\N_{\geq1}\}$. \item  $L\subseteq 0^*\cup 1^*$.
\end{enumerate}} 
\end{theorem}
%\todoscsInPlace{Reviewer 1 thinks that this should be combined with \autoref{prop:true-twins}.}
%\todoscs{Reviewer 1: The fact that treewidth is ill-suited for this goal is unsurprising, since treewidth is by design used for studying sparse classes. The natural parameter to try and use is clique width, the dense analogue of treewidth (or any of a number of parameters functionally equivalent to clique with, e.g. rank width, NLC-width, ... or generalising it, e.g. the newly defined flip-width). There is also twin width (presumably of particular interest here given that these classes are closed under adding some twins). The author(s) should investigate these parameters.}

\begin{proof}
$(1)\implies (2)$:  If all graphs in $\cG_L$ have treewidth at most~$t>0$ and some graph~$G$ in $\cG_L$ has treewidth at least one, then this graph contains an edge $\{u,v\}$. Now, as $\cG_L$ is closed under adding twins, we can first add $t$ twins~$u_1,\dots, u_{t}$ of~$u$, i.e., in the resulting graph $G'$, $V(G')=V(G)\cup \{u_1,\dots,u_{t}\}$ and $N_{G'}[u]=N_{G'}[u_i]$ for $i\in [t]$. Then, we add $t$ twins~$v_1,\dots, v_{t}$ of~$v$, i.e., in the resulting graph $G''$,  $V(G'')=V(G')\cup \{v_1,\dots,v_{t}\}$ and $N_{G''}[v]=N_{G''}[v_i]$ for $i\in [t]$. Now, $G''[\{u_i,v_i\mid i\in[t]\}\cup\{u,v\}]$ contains $K_{t+1,t+1}$ as a subgraph, a graph of treewidth $t+1$. By construction, $G''\in \cG_L$ according to \autoref{thm:twins}, contradicting our assumption that all graphs in~$\cG_L$ have treewidth at most~$t>0$. Therefore, if all graphs in $\cG_L$ have bounded treewidth, then all these graphs are null graphs. 
\\
$(2)\implies (1)$: trivial.
\\
$(2)\implies (3)$: If all graphs in $\cG_L$ have treewidth~$0$, then these graphs must all be null graphs. As $\cG_L\neq\emptyset$ and $\cG_L$ are hereditary, $N_1\in\cG_L$. As $\cG_L$ is closed under adding some twins and as adding true twins would introduce edges into graphs of $\cG_L$, we know that $\cG_L$ is in fact closed under adding false twins. \longversion{By \autoref{prop:false-twins}}\shortversion{Hence}, $\cG_L=\{N_n\mid n\in\N_{\geq1}\}$. 
\\
$(3)\implies (2)$:\shortversion{ trivial.}\longversion{ Clearly, $\{N_n\mid n\in\N_{\geq1}\}$ has treewidth~$0$.} 
\\
$(4)\implies (3)$: Extending \autoref{exa:represented-graphs-NullGraphs}, it is clear that $\cG_L=\{N_n\mid n\in\N_{\geq1}\}$ if $L\subseteq 0^*\cup 1^*$.
\\
$(3)\implies (4)$ by contraposition: If $L\not\subseteq 0^*\cup 1^*$, then there exists some pattern~$p\in L$ with $|p|_0>0$ and $|p|_1>0$. Hence, considering $V=\{0,1\}$, $G(L,p)=(V,\{\{0,1\}\})\simeq K_2\in \cG_L$, so that $\cG_L\neq \{N_n\mid n\in\N_{\geq1}\}$.%\qed 
\end{proof}
\end{toappendix}

\shortversion{\noindent}
The results of this section also have some interesting algorithmic consequences for decision problems\longversion{ that one might like to solve with our framework}.
\begin{theorem}\label{thm:deciding-bounded-treewidth-and-degeneracy}
%\shortversion{$(*)$}
There is an algorithm that, given a context-free grammar~$G$, decides if the language $L$ generated by~$G$ represents a graph class $\cG_L$ that has bounded treewidth or that has bounded degeneracy. %\todoReviewer{Reviewer 3: Is Theorem 5.8 true for any language closed under intersection with regular languages and with decidable non-emptiness? If yes, it is probably better to state it this way.}
\end{theorem}

%In a slightly different formulation, this theorem has been highlighted in the introduction as \autoref{thm:deciding-bounded-treewidth-and-degeneracy-intro}.

\begin{proof}
Consider first the treewidth question that is (more formally) the following one:
Given  a context-free grammar~$G$ (generating~$L$), does there exist some integer $k>0$ such that, for all graphs $H\in\cG_L$, $\text{tw}(H)\leq k$?
According to \autoref{thm:includingbigbiclique} and \autoref{exa:represented-graphs-NullGraphs}, this is equivalent to the question if $L\subseteq 0^*\cup 1^*$. This in turn is equivalent to asking if $L\cap \overline{0^*\cup 1^*}=\emptyset$. As $\overline{0^*\cup 1^*}$ is a \longversion{(even star-free)} regular language and context-free languages are effectively closed under intersection with regular languages, one can construct, from $G$, another context-free grammar $G'$ that generates the language $L\cap \overline{0^*\cup 1^*}$. Now, recall that the emptiness problem for context-free grammars is decidable (see \cite[Thm. 6.6]{HopUll79}), i.e., there exists an algorithm that decides, given~$G'$, if its language is empty. This algorithm clearly solves our original problem. 

Concerning degeneracy, observe that %the arguments given in \autoref{thm:d-degenerate} lead to the conclusion that 
a language-representable graph class of bounded degeneracy must contain null graphs only, so that we can take the same decision procedure as described for the bounded treewidth question also in this case for the question of bounded degeneracy.% \qed 
\end{proof}

\begin{remark}
\autoref{thm:deciding-bounded-treewidth-and-degeneracy} also holds for any language class that is effectively closed under intersection with regular languages and that has a decidable non-emptiness problem.
\end{remark} 

What about graph parameters\longversion{ (and related graph classes)} that are possibly small for dense graphs?
The most innocently looking of such parameters is probably \emph{neighborhood diversity} $\nd(G)$. This is the number of equivalence classes of the relation $u\sim_{\nd}v$ iff $N(u)\backslash\{v\}=N(v)\backslash\{u\}$. Obviously, $\nd(K_n)=1$ and $\nd(K_{n,n})=2$ if $n>1$.

\begin{theorem}\label{thm:no-bounded-nb-diversity-characterization}
    Let $k\in \mathbb{N}\setminus \{0,1\}$. Then there exists no $L \subseteq \{0,1\}^*$ such that $\mathcal{G}_L=\{G=(V,E)\mid \nd(G)=k\}$ or  $\mathcal{G}_L=\{G=(V,E)\mid \nd(G)\leq k\}$.
 \end{theorem}

\begin{proof}
    Assume there is a language $L \subseteq \{0,1\}^*$ such that $\mathcal{G}_L$ represents\longversion{ exactly} the graphs of neighborhood diversity~$(\leq)k$. 
    Define $G=(V,E)$ and $G'=(V',E')$\shortversion{:}\longversion{ with} 
    \begin{equation*}
    \begin{split}
        V &\coloneqq \{v_{i,j} \mid i\in [k],j\in [2]\},\quad V'\coloneqq \{v_{i,j}' \mid i\in [k],j\in [2]\}\\
        E&\coloneqq \{ \{v_{i,j}, v_{p,q}\} \mid  i,p\in [k], j,q\in [2], i\neq p \},\qquad E'\coloneqq \{ \{v'_{i,1}, v'_{i,2}\} \mid  i\in [k]\}
    \end{split}
    \end{equation*}
    Hence, $G'=kK_2$ and $G$ is the complement of $G'$. Clearly, the neighborhood diversity classes are given by $\{v_{i,1},v_{i,2}\}$ (for $G$) and $\{ v'_{i,1}, v'_{i,2} \}$ (for $G'$), with $i\in [k]$. Thus, $\nd(G)=\nd(G')=k$. Therefore, there are $w,w'$ such that $G=G(L,w)$ and $G'=G(L,w')$. Now we consider $G''=G(L,ww')$. Since $G,G'$ are induced subgraphs of $G''$, for $i,p\in [k], j,q\in [2]$ and $i \neq p$, $v_{i,j}$ together with $v_{p,q}$  and $v'_{i,j}$ together with $v'_{p,q}$ cannot be in one neighborhood diversity class. Hence $\nd(G'')\geq k$. Now assume there are $i,p\in [k]$ and $j,q\in [2]$ such that $v_{i,j}$ and $v'_{p,q}$ are in the same neighborhood diversity class. Then $\{v_{i,3-j},v'_{p,q}\}\notin E''$ but $\{v_{i,j},v'_{p,3-q}\}\in E''$. If $\{v_{i,j},v'_{p,q}\}\in E''$, $v_{i,j}$ and $v_{i,3-j}$ are not in one neighborhood diversity class. For $\{v_{i,j},v'_{p,q}\}\notin E''$, $v_{p,q}$ and $v_{p,1-q}$ are not in one neighborhood diversity class. Both cases imply $\nd(G'')>k$. %\qed
\end{proof}

%Under some technical conditions on~$L$, we can also rule out the possibility that a language-defined graph class only contains graphs of bounded neighborhood diversity. The idea of the construction is again based on twinning.

\begin{toappendix}
\begin{proposition}
    Let $\mathcal{G}$ be a graph class and $k\in \mathbb{N}\setminus \{0,1\}$ such that for all $G\in \mathcal{G}$, $\nd(G)\leq k$. If there is an $L\subseteq\{0,1\}^*$ with $\mathcal{G}=\mathcal{G}_L$ then for all $i\in \mathbb{N}$, $\langle 01 \rangle^i\subseteq L$ or $L \cap \langle 01 \rangle^i=\emptyset$. 
\end{proposition}

\begin{proof}
    Assume there is a $L\subseteq\{0,1\}^*$ with $\mathcal{G}=\mathcal{G}_L$ and an $i\in \mathbb{N}$, $\langle 01 \rangle^i \nsubseteq L$ and $L \cap \langle 01 \rangle^i \neq \emptyset$. Let $u_t\in L \cap \langle 01 \rangle^i$ and $u_f \in \langle 01 \rangle^i \setminus L$. 

    Let $w\in V^*$, $v \in V$ with $\vert w\vert=i $, $u\in \langle 01 \rangle^i$ and $x\notin V$. By definition, $u$ can be partitioned into $u_1\cdots u_i$ such that $u_j \in \langle 01 \rangle$ for all $j\in [i]$. Denote $w=w_1vw_2\cdots w_i v w_{i+1}$ such that $w_1,\ldots,w_{i+1}\in (V \setminus \{v\})^*$. Choose $x\notin V$. Define $f(w,v,x,u) \coloneqq w_1 g(u_1,v,x) w_2\cdots w_i g(u_i,v,x) w_{i+1}$ with 
    $$g(u_j,v,x) = \begin{cases}
        vx, & u_i=01\\
        xv, & u_i=10
    \end{cases}\quad.$$ 
    Consider $G=(V,E)=G(L,w)$ and $G'=(V\cup\{x\},E')= G(L,f(w,v,x,u))$.
    Let $p,q\in V\cup\{x\}$. If $x\notin \{p,q\}$ then 
    \begin{equation*}
        \begin{split}
            &h_{p,q}(f(w,v,x,u))\\
            ={}&h_{p,q}(w_1 g(u_1,v,x) w_2\cdots w_i g(u_i,v,x) w_{i+1})\\
            ={}& h_{p,q}(w_1)h_{p,q}( g(u_1,v,x))h_{p,q}( w_2)\cdots h_{p,q}(w_i)h_{p,q}( g(u_i,v,x) )h_{p,q}( w_{i+1})\\
            ={}& h_{p,q}(w_1)h_{p,q}(v) h_{p,q}( w_2)\cdots h_{p,q}(w_i)h_{p,q}( v )h_{p,q}( w_{i+1})\\
            ={}& h_{p,q}(w_1v w_2\cdots w_i v  w_{i+1})\\
            ={}& h_{p,q}(w).
        \end{split}
    \end{equation*}
    Hence, $\{p,q\}\in E$ \iffl $\{p,q\}\in E'$. Hence, $G'[V]=G$. 
    
    Next, consider $p=x$ and $q\neq v$. Hence,
    \begin{equation*}
        \begin{split}
            &h_{x,q}(f(w,v,x,u))\\
            ={}&h_{x,q}(w_1 g(u_1,v,x) w_2\cdots w_i g(u_i,v,x) w_{i+1})\\
            ={}& h_{x,q}(w_1)h_{x,q}( g(u_1,v,x))h_{x,q}( w_2)\cdots h_{x,q}(w_i)h_{x,q}( g(u_i,v,x) )h_{x,q}( w_{i+1})
        \end{split}
    \end{equation*}
    Since $w_i\in ((V\cup \{x\})\setminus \{v,x\})^*$, $h_{x,q}(w_i)=h_{v,q}(w_i)=1^{|w_i|_q}$. Moreover, $h_{x,q}( g(u_1,v,x)) = h_{v,q}(v) = 0 $, 
    \begin{equation*}
        \begin{split}
            &h_{x,q}(f(w,v,x,u))\\={}& h_{v,q}(w_1)h_{v,q}(v)h_{v,q}( w_2)\cdots h_{v,q}(w_i)h_{v,q}( v )h_{v,q}( w_{i+1})=h_{v,q}(w).
        \end{split}
    \end{equation*}
    Last, we consider  
    \begin{equation*}
        \begin{split}
            &h_{v,x}(f(w,v,x,u))\\=&h_{v,x}(w_1 g(u_1,v,x) w_2\cdots w_i g(u_i,v,x) w_{i+1})\\
            ={}& h_{v,x}(w_1)h_{v,x}( g(u_1,v,x))h_{v,x}( w_2)\cdots h_{v,x}(w_i)h_{v,x}( g(u_i,v,x) )h_{v,x}( w_{i+1})\\
            ={}& h_{v,x}( g(u_1,v,x))\cdots h_{v,x}( g(u_i,v,x) )\\
            ={}& u_1\cdots u_i =u.
        \end{split}
    \end{equation*}
    Hence, $\{v,x\}\in E$ \iffl $u \in L$. Thus, $x$ is a twin of $v$. $x$ is a true twin \iffl $u\in L$.
    One can hence view the operation of adding~$x$ as a specific form of twinning.

    Define $G_r\coloneqq([2r],E_r)$ for $r\in \mathbb{N}$ with
    $$E_r\coloneqq \{\{2i-1,2j-1\}\mid i,j\in [r]\}\cup \{\{2i-1,2j\}\mid i\in [r],j \in [i]\}.$$ 
    For $i,j\in [r]$ with $i<j$, $2i+1 \in N(2j)\setminus N(2i)$. Therefore, $\nd(G_r)\geq r$. Now we want to prove that $G_r\in \mathcal{G}_L$. $G_1=G(L,f(1^i,1,2,u_t))$ is a $K_2$. Define $w_1=f(1^i,1,2,u_t)$ and  $w_{r+1}\coloneqq f(f(w_{r},2r,2r+1,u_f),2r+1,2(r+1),u_t)$, $G_{r}'=G(L,w_r)$ for $r\in \mathbb{N}$. 

    We want to prove by induction that $G_r=G'_r$. For $r=1$, both are isomorphic to a $K_2$. Let $r\in \mathbb{N}$ such that $G_r=G'_r$ and we consider now $G_{r+1}$ and $G_{r+1}$. Clearly, both have the same vertex set. Furthermore, $G_{r+1}[[2r]]=G_r=G_r'=G'_{r+1}[[2r]]$. By our argumentation above and the definition of $u_t$ and $u_f$, $2r+1$ is a false twin of $2r$ on $G'_{r+1}[2r+1]$ and $2(r+1)$ is a true twin of $2r+1$ on $G'_{r+1}$. Thus, \begin{equation*}
        \begin{split}
            N_{G'_{r+1}}[2r+2]=N_{G'_{r+1}}[2r+1]=N_{G'_{r}}(2r) \cup \{2r+1,2(r+1)\}\\ = \{2i-1,2i\mid i\in [r+1]\}=N_{G_{r+1}}[2r+1]=N_{G_{r+1}}[2r+2]
        \end{split}
    \end{equation*}
    Hence, $G_r=G'_r$ for all $r\in \mathbb{N}$ and the neighborhood diversity in $\G_L$ is not bounded, contradicting that $\cG=\cG_L$.%\qed 
\end{proof}
\end{toappendix}

\section{Multi-Word-Representability: Some First Ideas}

In the original definition, a graph~$G$ is $k$-word-representable if, given $k$ words $W=\{w_1,\dots,w_k\}$, $G$ has an edge $uv$ if $u$ and $v$ alternate in one of the words in~$W$. In other words, $G$ is the union of the graphs $G_1,\dots,G_k$ described by $w_1,\dots, w_k$.
Similarly, given a 0-1-symmetric language $L$ and $k$ words $W=\{w_1,\dots,w_k\}$, we can define $G(L,W)$ as the union of the graphs $G(L,w_1)$, \dots, $G(L,w_k)$.
This could define the notion of \emph{generalized multi-word-representability}.
We propose to study this notion more in depth, as it obviously generalizes both multi-word-representability and language-representability.

\begin{example}
We revisit $G(\langle0011,0101\rangle,\mathtt{aebcdcadbffe})$ from \autoref{fig:graphs-for-one-word-with-different-L}. Observe that this graph equals the union of $G(\langle 0101\rangle,\mathtt{aebcdcadbffe})$ and $G(\langle 0101\rangle,\mathtt{fabcdfdcbaee})$. Hence, it is 2-word-representable with respect to $\langle 0101\rangle$. However, this is indeed a circle graph, as testified by $\mathtt{acfdcbeaefbcd}$.
\end{example}

%In order to prove the main result of this section, we need another word-combinatorial result.
\begin{proposition}\label{prop:kl-uniform}
Let $w\in \Sigma^*$. Then, $w\in \left(\Sigma^{\text{$\ell$-uni}}\right)^{k}$ (i.e., $w$ is the concatenation of $k$ $\ell$-uniform words) \iffl $h_{\ta,\tb}(w)\in \left(\{0,1\}^{\text{$\ell$-uni}}\right)^{k}$ for all $\{\ta,\tb\}\in \binom{\Sigma}{2}$.
\end{proposition}

\begin{proof}
If $w=w_1\cdots w_k$, with $w_i$ being $\ell$-uniform for each $1\leq i\leq k$, then $h_{\ta,\tb}(w)=h_{\ta,\tb}(w_1)\cdots h_{\ta,\tb}(w_k)$, with each $h_{\ta,\tb}(w_i)$  being $\ell$-uniform for any $\{\ta,\tb\}\in \binom{\Sigma}{2}$, as all $h_{\ta,\tb}$ are morphisms.
Conversely, assume $h_{\ta,\tb}(w)\in \left(\{0,1\}^{\text{$\ell$-uni}}\right)^{k}$ for all $\{\ta,\tb\}\in \binom{\Sigma}{2}$. Clearly, $w$ must be $k\cdot \ell$-uniform. If $w\notin \left(\Sigma^{\text{$\ell$-uni}}\right)^{k}$, then there must be some $i\in [k]$ and $\ta\neq\tb$ such that within~$w$, the position $p_\ta(i\ell)$ of the  $(i\ell)^{\text{th}}$ occurrence of~$\ta$ is after the position $p_\tb(i\ell+1)$   of the  $(i\ell+1)^{\text{th}}$ occurrence of~$\tb$. But then $p_0(i\ell)>p_1(i\ell+1)$ in $h_{\ta,\tb}(w)$, which contradicts $h_{\ta,\tb}(w)\in \left(\{0,1\}^{\text{$\ell$-uni}}\right)^{k}$.
\end{proof}

\noindent
In  this paper, we would only like to mention the following formal-language related main result.

\begin{theorem}
Let  $\ell>1$ and let $L\subseteq\{0,1\}^{\text{$\ell$-uni}}$ be 0-1-symmetric and assume that there is a $j \in [\ell]\cup\{0\}$ with $0^{\ell - j} 1^\ell 0^j \notin L$. Then for each $k>1$, there is a $k\cdot \ell$-uniform language $L_k$ such that the graph~$G$ is $k$-word-representable via~$L$ \iffl $G$ is $L_k$-representable.
\end{theorem}

\begin{proof}Fix  $\ell>1$ and some 0-1-symmetric $L\subseteq\{0,1\}^{\text{$\ell$-uni}}$. Choose $0\leq j\leq \ell$ with $0^{\ell - j} 1^\ell 0^j \notin L$.
Let $k>1$ be arbitrary. 
For an alphabet $\Sigma$ and for $i\in [k]$, let $p_i: \Sigma^{\text{$k\ell$-uni}} \to \Sigma^{\text{$\ell$-uni}}$ be the projection to the $((i-1)\ell+1)^{\text{th}}$ through the $(i\ell)^{\text{th}}$ occurrence of each letter; e.g., for $k = 3$ and $\ell = 2$, $p_3(\mathtt{{\color{gray}aacaa}a{\color{gray}bccbcb}ca{\color{gray}b}bbc}) = \mathtt{acabbc}$. Let $L_k \coloneqq \{ w \in \{0,1\}^{\text{$k\ell$-uni}} \mid \exists i \in [k]: p_i(w) \in L \} $.
Let $G = (V,E)$ be a graph.

First, assume that $G$ is $k$-word-representable via $L$, i.e., there exist $w_1, \dots, w_k \in V^*$ such that $G = \bigcup_{i \in [k]} G(L,w_i)$. For each $i \in [k]$, let $U_i \coloneqq \{ \ta \in V \mid |w_i|_\ta = \ell\}$ and $w_i' \coloneqq \prod_{\ta \in V \setminus U_{i}} a^{\ell -j} \cdot h_{U_i}(w_i) \cdot (\prod_{\ta \in V \setminus U_{i}} a^j )^R$. Observe that $G(L,w_i)=G(L,w_i')$ and that $w_i'\in \Sigma^{\text{$\ell$-uni}}$ for each $i\in [k]$. Let $w \coloneqq \prod_{i \in [k]} w_i'$. Let $\{\ta,\tb\} \in \binom{V}{2}$. 
\begin{align*}
\{\ta,\tb\} \in E 
& \Leftrightarrow \exists i \in [k]: \{\ta,\tb\} \in E(G(L,w_i)) \\
& \Leftrightarrow \exists i \in [k]: h_{\ta,\tb}(w_i) \in L \Leftrightarrow \exists i \in [k]: h_{\ta,\tb}(w_i') \in L\\
& \Leftrightarrow \exists i \in [k]: h_{\ta,\tb}(w_i') = h_{\ta,\tb}(p_i(w)) = p_i(h_{\ta,\tb}(w)) \in L \\
& \Leftrightarrow h_{\ta,\tb}(w) \in L_k 
\end{align*}
Notice that \autoref{prop:kl-uniform} implies $h_{\ta,\tb}(p_i(w)) = p_i(h_{\ta,\tb}(w))$.
Therefore, $G = G(L_k,w)$. 

For the other direction, assume $G \in \mathcal{G}_{L_k}$, i.e., there is a word $w \in V^*$ such that $G = G(L_k, w)$. Let $U \coloneqq \{ \ta \in V \mid |w|_\ta = \ell k\} $. For each $i \in [k]$, let $w_i \coloneqq p_i(h_U(w)) \cdot \prod_{\ta \in V \setminus U} \ta^{\ell -1}$. For each $\ta \in U$ and $\tb \in V \setminus \{\ta\}$, $\{ \ta, \tb \} \notin E$ and for each $i \in [k]$, $\{ \ta, \tb \} \notin E(G(L,w_i))$. 
For each $\{\ta,\tb\} \in \binom{V\setminus U}{2}$, 
\begin{align*}
\{\ta,\tb\} \in E 
& \Leftrightarrow h_{\ta,\tb}(w) = h_{\ta,\tb}(h_U(w)) \in L_k \\
& \Leftrightarrow \exists i \in [k]: p_i(h_{\ta,\tb}(h_U(w))) = h_{\ta,\tb}(p_i(h_U(w))) = h_{\ta,\tb}(w_i) \in L \\
& \Leftrightarrow \exists i \in [k]: \{\ta,\tb\} \in E(G(L,w_i)).
\end{align*}
(Again, $h_{\ta,\tb}$ and $p_i$ commute by \autoref{prop:kl-uniform}.) 
Therefore, $G = \bigcup_{i \in [k]} G(L,w_i)$. 
\end{proof}

\noindent
With the same arguments that have led to \autoref{thm:hereditary}, one can prove:
\begin{theorem}\label{thm:hereditary-multiword}
For any 0-1-symmetric language $L$ and any $k\geq 1$, the class of graphs $\cG_{L,k}$ (that are $k$-word-representable via~$L$) is  hereditary.
\end{theorem}

\begin{toappendix}
Kenkireth \emph{et al.}~\cite{KenSalSas2025} introduced the \emph{multi-word representation number} $\mu(G)$ of a graph~$G$. We propose to generalize it towards the \emph{covering number} $\mu_L(G)$ ({w.r.t.}\longversion{with respect to}~$L$) as the smallest number~$k$ of words such that $G$ is $k$-word-representable via~$L$. With $L=\langle 001,01\rangle$, by \autoref{exa:threshold}, $\mu_L(G)$ is the well-known \emph{threshold covering number}, also known as  \emph{threshold dimension}\footnote{Attention: there are other definitions for the notion of threshold dimension in the context of metric dimension~\cite{MolMurOel2020}.}, which motivates our definition. For this notion, we refer to \cite{ChvHam77,CozLei84,CozLei87,CozHal91,ErdOrdZal87,Yan82}; in particular, the last paper showed that it is \NP-complete to determine if the threshold covering number of a given graph is at most~3. (The same question for~2 seems to be open while it is polynomially-time solvable for~1.) Of course, similar complexity questions can be studied for any language~$L$. To the best of our knowledge, nothing has been explored in this area so far. In this sense and following previous naming conventions in graph theory, the \emph{multi-word representation number} introduced in~\cite{KenSalSas2025}  might be rather called the \emph{covering number with respect to word representability}. Currently, it is even open if there exists any graph~$G$ with $\mu_{L_{\classical}}(G)\geq 3$. This makes \NP-hardness difficult to prove even for the question if $\mu_{L_{\classical}}(G)\leq k$, given a graph~$G$ and an integer $k\geq 2$ (that might be also fixed). However, for $k=1$, this question is known to be \NP-hard even on triangle-free graphs by~\cite{HalKitPya2016} as this is just the question if a given graph is word-representable.
\end{toappendix}

\section{Conclusion}
In this work, we further investigated the 
\longversion{new notion of $L$-representability -}\shortversion{idea of} representing graph classes by formal languages. \longversion{We have shown that e}\shortversion{E}ven with small finite languages, we can represent powerful classes like interval or circle graphs. We fully characterized the graph classes described by $1$-, $(1,2)$- and $2$-uniform languages. For most other languages, this remains \longversion{an open question}\shortversion{open}. We refer to \cite{FenFFKS2026} for the study of\longversion{ many} other finite languages, and to \cite{FerFMS2026a} \longversion{concerning}\shortversion{for} infinite languages. 
%In contrast to the classical representation, 
Our setting implies some closure properties, e.g., if a graph class is not closed under twinning, there is no language representing it and all\longversion{ graph} classes described by our model are hereditary. Those properties give insights for identifying the graph classes that can be described by a given binary language. We have also extended our approach further to multi-word-representability and showed that under some conditions, such a notion would be also captured by word-representability (\shortversion{w.r.t.}\longversion{with respect to} another language).

\bibliographystyle{eptcs}
\bibliography{ab,hen,wrg-addendum}

%\inappendixtrue
%\appendix

\end{document}